\documentclass[a4paper]{article}
\usepackage[utf8]{inputenc}
\usepackage[english]{babel} % English language/hyphenation
\usepackage{amsmath,amssymb,amsthm} % Math packages
\usepackage[margin=2cm]{geometry}
\usepackage{esint}
\usepackage[colorlinks]{hyperref}
\usepackage{enumerate}
\usepackage{framed,comment}
\usepackage{upgreek}
\usepackage{pgfplots}
\usepackage{float}
\usepackage{tikz,tikz-cd}
\usepackage{rotating}
\usepackage{todonotes}

\newtheorem{theorem}{Theorem}[section]
\newtheorem{corollary}{Corollary}[theorem]
\newtheorem{definition}{Definition}[section]

\newtheorem{lemma}[theorem]{Lemma}

\newtheorem{remark}[theorem]{Remark}

\usepackage{parskip}
\numberwithin{equation}{section}

\def\cd{\cdot}

\def\zh{\mathbf{\hat{z}}}

\def\bx{{\bf x}}

\def\bR{{\bf R}}

\def\bu{{\bf u}}

\newcommand{\sym}[1]{\boldsymbol{#1}}

\pgfplotsset{compat=1.14}

\begin{document}
\title{Stochastic mesoscale circulation dynamics in the thermal ocean \\
\bigskip\Large
Darryl D.Holm\quad\hbox{and}\quad Erwin Luesink\quad\hbox{and}\quad Wei Pan \\ \bigskip\large
Department of Mathematics, Imperial College London SW7 2AZ, UK
\\ \bigskip\normalsize
d.holm@ic.ac.uk and e.luesink16@ic.ac.uk and wpan1@ic.ac.uk} 
\date{}                                           % Activate to display a given date or no date

\maketitle
%\section{}
%\subsection{}

\makeatother

\noindent
\begin{abstract}

In analogy with similar effects in adiabatic compressible fluid dynamics, the effects of buoyancy gradients on incompressible stratified flows are said to be `thermal'. 
The thermal rotating shallow water (TRSW) model equations contain three small nondimensional parameters. These are the Rossby number, the Froude number and the buoyancy parameter. 
Asymptotic expansion of the TRSW model equations in these three small parameters leads to the deterministic thermal versions of the Salmon's L1 (TL1) model and the thermal quasi-geostrophic (TQG) model, upon expanding in the neighbourhood of thermal quasi-geostrophic balance among the flow velocity and the gradients of free surface elevation and buoyancy. The linear instability of TQG at high wave number tends to create circulation at small scales. Such a high wave number instability could be unresolvable in many computational simulations, but its presence at small scales may contribute significantly to fluid transport at resolvable scales. Sometimes such effects are modelled via `stochastic backscatter of kinetic energy'. Here we try another approach. Namely, we model `stochastic transport' in the hierarchy of models TRSW/TL1/TQG. The models are derived via the approach of  \emph{stochastic advection by Lie transport} (SALT) as obtained from a recently introduced stochastic version of the Euler--Poincar\'e variational principle. We also indicate the potential next steps for applying these models in uncertainty quantification and data assimilation of the rapid, high wavenumber effects of buoyancy fronts at these three levels of description by using the data-driven stochastic parametrisation algorithms derived previously using the SALT approach.

\end{abstract}
%%%%%%%%%%%%%%%%%%%%%%%%%%%%%%%%%

%\tableofcontents

\section{Introduction}

In this paper we are dealing with the thermal rotating shallow water (TRSW) equations, which can be regarded as the vertically averaged version of the primitive equations with a buoyancy variable \cite{zeitlin2018geophysical}.

In the balanced 2D model hierarchy of TRSW, TL1 and TQG, we are investigating a certain stochastic model of potential vorticity dynamics as a basis for stochastic parametrisation of the dynamical creation of unresolved degrees of freedom in computational simulations of upper ocean dynamics. Specifically, we have chosen the SALT (Stochastic Advection by Lie Transport) algorithm introduced in \cite{holm2019stochastic} and applied in \cite{cotter2018modelling, cotter2019numerically} as our modelling approach. The SALT approach preserves the Kelvin circulation theorem and an infinite family of integral conservation laws.
The goal of the SALT algorithm is to quantify the uncertainty in the process of upscaling, or coarse graining of either observed or synthetic data at fine scales, for use in computational simulations at coarser scales. The present work prepares us to take the next step from (ii) to (iii) in the well-known path of discovery in oceanography, weather prediction and climate science, which is
\begin{enumerate}[(i)]
\item
driven by large datasets and new methods for its analysis;
\item
informed by rigorous mathematical derivations and analyses of stochastic geophysical fluid equations;
\item
quantified using computer simulations, evaluated for uncertainty, variability and model error;
\item
optimized by cutting edge data assimilation techniques, then
\item
compared with new observation datasets to determine what further analysis and improvements will be
needed.
\end{enumerate}

The objective in applying the SALT algorithm to coarse grained simulations is to answer the following question, 
enunciated in \cite{cotter2018modelling, cotter2019numerically}:
``How can one use computationally simulated surrogate data at highly resolved scales, in combination
with the mathematics of stochastic processes in nonlinear dynamical systems, to estimate and model
the effects on the simulated variability at much coarser scales of the computationally unresolvable, small,
rapid, scales of motion at the finer scales?''
The present paper will lay the theoretical foundations for addressing this question in the 2D context of the thermal rotating shallow water (TRSW) model and its balanced thermal quasi-geostrophic (TQG) model. Our eventual goal is to apply the SALT algorithm to calibrate our stochastic models for assimilating data, e.g., from satellite observations of the cascade in the upper ocean dynamics of horizontally circulating structures to smaller scales, as shown in Figure \ref{Fig:globcurrent1} below.
%%%%%%%%%%%%%%%%%%%%%%%%%
\begin{figure}[h!]
\begin{center}
\includegraphics[width=0.8\textwidth]{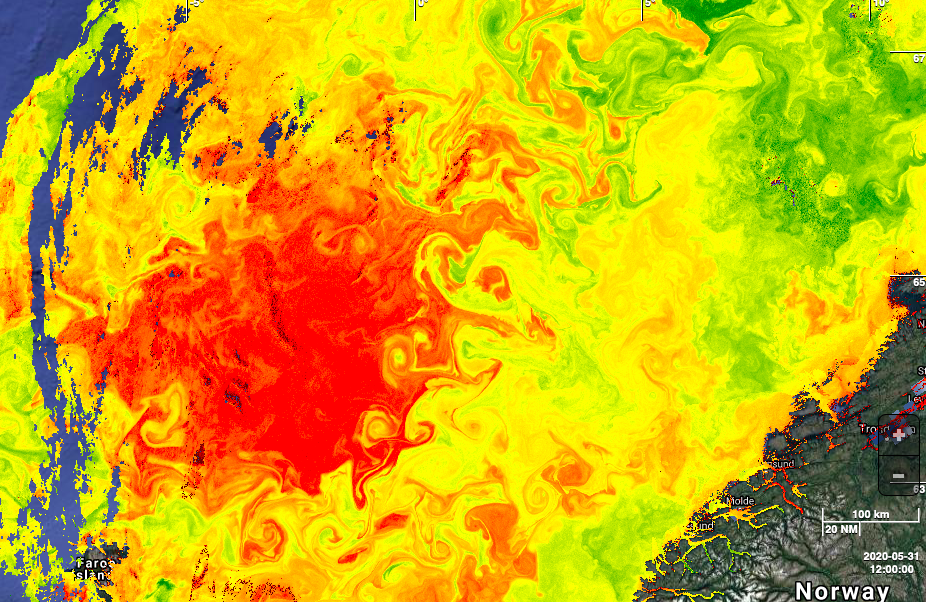}
\caption{\footnotesize This image of the Lofoten Vortex, courtesy of 
\href{https://ovl.oceandatalab.com/?date=1590926400268\&timespan=1d\&zoom=6\&extent=-2750822.3231634\%2C8280531.8867367\%2C1945468.6940242\%2C10572419.74252\&products=3857_Sentinel-2_RGB\%2C3857_AMSR_sea_ice_concentration\%2C3857_Sentinel-3A_OLCI_Chlorophyll_a_oc4me\%2C3857_Sentinel-3B_OLCI_Chlorophyll_a_oc4me\&opacity=100\%2C100\%2C100\%2C100\&stackLevel=95\%2C10\%2C84\%2C84.01}{https://ovl.oceandatalab.com/}
 illustrates the configurations of submesoscale currents obtained from ESA Sentinel-3 OLCI instrument observation of chlorophyll on the surface of the Norwegian Sea in the Lofoten Basin, near the Faroe Islands. Fueled by warm saline Atlantic waters crossing the Norwegian Sea, the Lofoten Basin is a major reservoir of heat whose buoyancy gradient interacts with the bathymetry gradient to sustain  a large anticyclonic vortex exhibiting intense mesoscale and submesoscale activity.  In particular, the figure shows many of the features of submesoscale currents surveyed in \cite{mcwilliams2019survey}. High resolution ($4km$) computational simulations of the Lofoten Vortex have recently discovered that its time-mean circulation is primarily barotropic, \cite{volkov2015formation}, thereby making the Lofoten Vortex a reasonable candidate for investigation using vertically averaged dynamics such as the TQG approach. For more information about the Lofoten Vortex, see, e.g., \cite{filyushkin2018evolution,bashmachnikov2017vertical, bashmachnikoveddies, bashmachnikov2018pattern, fedorov2020interaction} and references therein. 
}\label{Fig:globcurrent1}
\end{center}
\end{figure}

%%%%%%%%%%%%%%%%%%%%%%%%%
\begin{figure}[h!]
\begin{center}
\includegraphics[width=.465\textwidth]{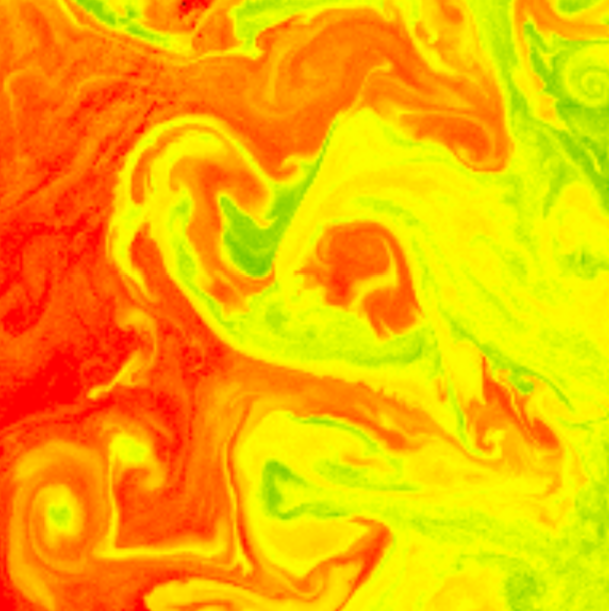}
\quad
\includegraphics[width=.4\textwidth]{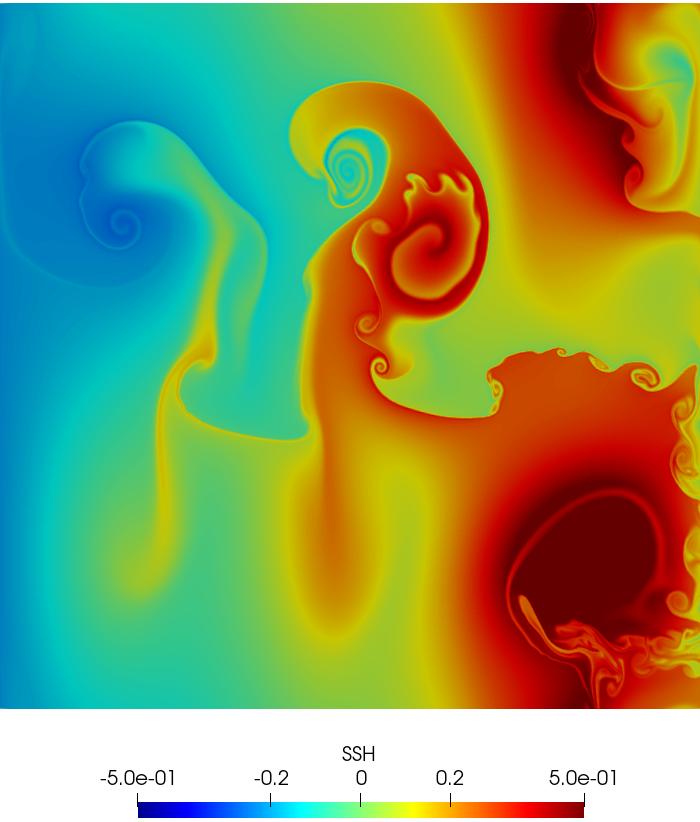}
\end{center}
\caption{\footnotesize An excerpt of the Lofoten Vortex satellite data from the centre of Figure \ref{Fig:globcurrent1} is compared with an image of sea surface height taken from an evolutionary computational simulation of deterministic TQG. The latter image illustrates the result of the creation of a buoyancy front from an initial state and its subsequent emergent cascade of circulation. More details of this computation are given in the caption of Figure \ref{Fig:CourtesyofWeiPan2}. The structural features of the two figures are similar in appearance. }\label{Fig:CourtesyofWeiPan-Intro}
\end{figure}

\paragraph{The TRSW equations.}
As shown in figure \ref{Fig:Tree-GFDapprox}, the TRSW equations arise in a series of nested approximations leading from the 3D Euler equations for inhomogeneous, stratified, rotating incompressible fluids, first to the 3D Euler--Boussinesq equations for small stratification, then to the rotating thermal Green-Naghdi equations, which were derived and investigated in \cite{holm2019stochastic}. Upon further neglecting the nonhydrostatic pressure effects which are present in the Green-Nahgdi model, the TRSW model is obtained. The TRSW equations (which were first called the IL$^0$ model in \cite{ripa1995low}) and their quasi-geostrophic approximation, the TQG equations, comprise standard models of thermal effects in GFD, as reviewed, e.g., in \cite{brocker2018mathematics, zeitlin2018geophysical}. The TQG equations are discussed in the GFD literature by Warneford and Dellar \cite{warneford2013quasi}, for example, following earlier work by Ripa  \cite{ripa1993conservation, ripa1995low} and \cite{ripa1999validity}. In fact, the TRSW and TQG equations and their high wave-number instabilities have been rederived several times, as recounted in \cite{zeitlin2018geophysical}. A multilayer extension of the shallow water model with stratification and shear can be found in \cite{beronvera2020multilayer}, which includes a historical background on the developments of the thermal rotating shallow water model. In this paper, we will derive the SALT \emph{stochastic versions} of TRSW and TQG, as well as TL1, which is  a thermal version of an intermediate theory known in the GFD literature as Salmon's L1 model \cite{salmon1983practical}.  In deriving the SALT versions of these  deterministic equations, we will follow the geometric approach of \cite{holm2015variational} which is based on Hamilton's variational principle for Eulerian fluid flows \cite{holm1998euler}. 

\begin{figure}[H]
%%%%%%%%%%%%%%%%%%%%%%%%%%%%%%%%%%%%%%%%%%%%%%
\footnotesize
\centering
\begin{tikzcd}
[row sep = 3em, 
column sep=2em, 
cells = {nodes={top color=green!20, bottom color=blue!20,draw=black}},
arrows = {draw = black, rightarrow, line width = .01cm}]
& 
&
%1st row 3rd column
	{\begin{matrix}
	\text{3D Euler equations for}\\ \text{for stratified, rotating}
	\\ 
	\text{incompressible fluids}
	\end{matrix}} 
	\arrow[d,"{\begin{matrix} \footnotesize \text{small buoyancy}\\ \text{stratification}\end{matrix}}"', 
	shorten <= 1mm, shorten >= 1mm] 
&
\\
&	 
	{\begin{matrix} \text{Primitive equations} 
	\\ 
	\text{with stratification} 
	\end{matrix}} 
	\arrow[d,
		"{\begin{matrix} 
		\text{vertical average,}
		\\ 
		\text{small wave amplitude}
		\end{matrix}}"', 
	shorten <= 1mm, shorten >= 1mm]
& 
	{\begin{matrix} 
	\text{3D Euler--Boussinesq equations}
	\\ 
	\text{for stratified, rotating,}
	\\ 
	\text{incompressible fluids}
	\end{matrix}} 
	\arrow[dr, 
		"{\begin{matrix} \text{vertical average,}
		\\ 
		\text{very small wave amplitude}\end{matrix}}", 
		shorten <= 4mm, shorten >= 4mm] 
		\arrow[d, "{\begin{matrix} \text{vertical average,}
		\\ \text{small wave amplitude}\end{matrix}}"', 
		shorten <= 1mm, shorten >= 1mm] 
	\arrow[l,
		"\begin{tabular}{l}
		\tiny \text{hydrostatic} 
		\\ 
		\tiny \text{approx.} 
		\\ 
		\text{ } \end{tabular}"', 
		shorten <= 1mm, shorten >= 1mm] 
&
\\
&
	{\begin{matrix} 
	\text{Thermal rotating}
	\\ 
	\text{shallow water} 
	\end{matrix}}
	\arrow[dd,
	"\text{restrict to 1D}", 
	shorten <= 1mm, shorten >= 1mm] 
	\arrow[dl,
	"\begin{tabular}{l} 
	\tiny \text{asymptotic expansion around} 
	\\
	\tiny \text{thermal geostrophic balance}
	\end{tabular}"', 
	shorten < =4mm, shorten >= 4mm]
& 
	{\begin{matrix} 
	\text{Thermal rotating}
	\\ 
	\text{Green--Naghdi equations} 
	\end{matrix}} 
	\arrow[dd, "\text{restrict to 1D}"', 
	shorten <= 1mm, shorten >= 1mm] 
	\arrow[l,
	"\begin{tabular}{l} 
	\\ 
	\tiny \text{hydrostatic}
	\\ 
	\tiny \text{approx.} 
	\end{tabular}", 
	shorten <= 1mm, shorten >= 1mm] 
	\arrow[r,
	"\begin{tabular}{l} \text{ } 
	\\ 
	\tiny \text{rigid lid} 
	\end{tabular}"', 
	shorten <= 1mm, shorten >= 1mm]
& 
	{\begin{matrix} 
	\text{Thermal rotating}
	\\ 
	\text{Great Lake equations} 
	\end{matrix}}
	\quad 
	\arrow[d, 
	"\begin{tabular}{l} 
	\tiny \text{hydrostatic} 
	\\ 
	\tiny \text{approx.} 
	\end{tabular}"', 
	shorten <= 1mm, shorten >= 1mm] 
\\
	{\begin{matrix} 
	\text{Thermal L1 model,} 
	\\ 
	\text{Thermal quasi-geostrophic} 
	\\ 
	\text{ (TQG) equations} 
	\\ 
	\text{ } 
	\\ 
	\text{TQG reduces to QG} 
	\\ 
	\text{for constant buoyancy.} 
	\end{matrix}}
%	\arrow[d,
%	"\begin{tabular}{l} 
%	\tiny \text{Constant } 
%	\\ 
%	\tiny \text{buoyancy} 
%	\end{tabular}"', 
%	shorten <= 1mm, shorten >= 1mm] 	
&
&
& 
	{\begin{matrix} 
	\text{Thermal rotating}
	\\ 
	\text{Lake equations}
 	\end{matrix}}
\\
%	{\begin{matrix} 
%	\text{Quasi-geostrophic} 
%	\\ 
%	\text{ (QG) equations} 
%	\end{matrix}}
& 	
	{\begin{matrix} 
	\text{TS-V, S-V, Burgers}
	\\ \text{Bore solutions}
	\end{matrix}}
& 	
	{\begin{matrix} 
	\text{TCH3, CH2, CH, KdV}
	\\ 
	\text{Soliton solutions}
	\end{matrix}} 
	\arrow[l,
	"\begin{tabular}{l} 
	\\ \tiny \text{hydrostatic} 
	\\ \tiny \text{approx.} 
	\end{tabular}", 
	shorten <= 1mm, shorten >= 1mm] 
&
\end{tikzcd}
\caption{\footnotesize {This tree-diagram describes how the various asymptotic expansions, hydrostatic assumptions and vertical averages lead to a plethora of intimately related GFD approximations. The present work concentrates on the paths leading from the thermal rotating shallow water (TRSW) equations to the thermal quasi-geostrophic (TQG) equations via the the thermal version of Salmon's L1 model.}}
\label{Fig:Tree-GFDapprox}
\end{figure}
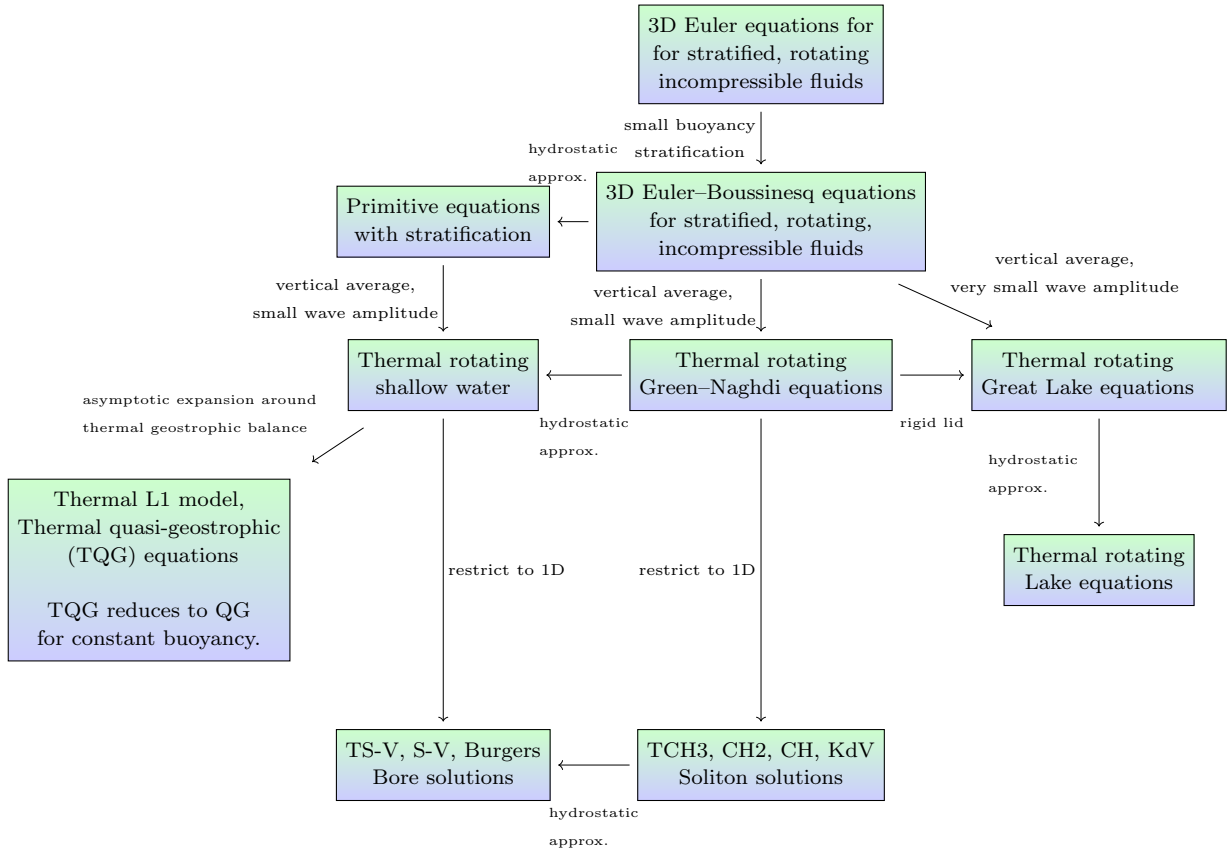
\normalsize
%%%%%%%%%%%%%%%%%%%%%%%%%%%%%%%%%%%%%%%%%%%%%%

The TRSW, TL1 and TQG models separate the wave and current aspects of their flows into gravity waves and Rossby waves on the free surface, and fluid circulation in the region between the free surface and the bottom topography. The Kelvin circulation theorems for TRSW, TL1 and TQG show that horizontal gradients of the buoyancy in the fluid region (e.g., at thermal fronts) can couple to either the elevation gradient at the upper interface, or to the bathymetry gradient at the lower interface. Namely, horizontal circulation of the fluid is produced whenever either of the gradients at the upper and lower interfaces are misaligned with the horizontal gradient of the buoyancy. Thus, both waves on the surface and variations of the bottom topography can create horizontal fluid circulation when the buoyancy is spatially inhomogeneous. 

The generation of submesoscale circulations involves a wide range of time scales, as well as many couplings among the various degrees of freedom and the boundaries. The `irreducible imprecision' of numerical simulations \cite{mcwilliams2007irreducible} and the sparsity of observed data in both space and time produce uncertainty in forecasts of rapid, high-wavenumber GFD processes and thereby present a `grand challenge' for data assimilation. 

In preparation for meeting this challenge, the present paper develops the stochastic variational principles for the TRSW, TL1 and TQG models. This mathematical foundation is needed in applying the SALT  (Stochastic Advection by Lie Transport) approach to the derivation of stochastic fluid equations which preserve the geometric structure of fluid dynamics \cite{holm2015variational}.  The physical effects of stochasticity in the SALT approach for deriving the stochastic TRSW, TL1 and TQG models are revealed in their Kelvin circulation theorems. Namely, the corresponding material loops defining the circulation integrals for these models are shown to move along stochastic Lagrangian paths.  The  motivations and recent results of applications of the SALT approach for uncertainty quantification and for data assimilation are laid out in \cite{cotter2018modelling,cotter2019numerically}. \medskip

\begin{remark}[Separation of scales behaviour in TQG versus QG]
High Reynolds-number two-dimensional Navier-Stokes turbulence relaxes through a combination of vortex merger and filamentation, in which energy tends to flow to large scales and enstrophy tends to flow and be dissipated on fine scales. The combined conservation of two integral conservation laws of energy and enstrophy produces the famous \emph{inverse cascade} \cite{kraichnan1967inertial}, which also occurs in QG turbulence. However, TQG turbulence is different. First, TQG does not preserve an enstrophy because there are two active degrees of freedom in TQG, both the momentum and the buoyancy. Second, TQG possesses high-wavenumber instabilities. That is, linear analysis reveals that high wave numbers in a TQG flow can suddenly become unstable. Turbulence tends to mimic the properties and locations of its energy source. Consequently, numerical simulations of the nonlinear processes in TQG reveal the nonlinear sudden creation of coherent structures at the scale of highest growth rate at the onset of the high-wavenumber, small-lengthscale instabilities. Thus, unlike QG turbulence, a scale-separation exists in TQG turbulence. Therefore, one can expect TQG dynamics to mimic QG dynamics at low wavenumbers and to transform into another type of motion when the flow enters the onset of the high-wavenumber instability for TQG.  When this happens, the turbulence will be limited to the scale of the wavelength at which the highest growth-rate occurs, which might be significantly smaller than the scale of the entire domain, see figure \ref{Fig:CourtesyofWeiPan2}. This means TQG possesses a scale separation in its solution behaviour which can be radically different from the solution behaviour of other models in the QG family. In particular, one can expect that TQG will require different approaches from QG in the methods of its analysis and parameterisation. Indeed, the misalignment of gradients of bathymetry with gradients of buoyancy can create circulation even when there is no initial flow. Because of the scale gap in its solution behaviour, one may expect that TQG may require different methods of analysis,  parameterisation and numerical implementation  in the same computational simulation of a given TQG flow. 
\end{remark}

\paragraph{Content of the paper}
\begin{enumerate}

\item In section \ref{DTRSW-sec} we review the deterministic TRSW model by re-deriving its equations in the Euler--Poincar\'e variational framework of \cite{holm1998euler}. In the Euler--Poincar\'e framework, we prove the Kelvin--Noether circulation theorem and discuss steady solution properties of the deterministic TRSW equations. This derivation of the TRSW equations with stochastic advection by Lie transport (SALT) is intended to be the mathematical foundation for a  systematic means of introducing data-driven  parametrisations of stochastic transport for uncertainty quantification and data assimilation for upper ocean dynamics.  This type of uncertainty quantification and data assimilation has already been accomplished using this approach for the 2D Euler equations in a square domain with fixed boundaries and the 2-layer QG equations in a periodic channel, in \cite{cotter2018modelling} and \cite{cotter2019numerically}, respectively.

\item In section \ref{sec:detEliassen} we discuss the deterministic thermal Eliassen approximation, or TL1 model, of TRSW, as derived from a combination of the Euler--Poincar\'e variational approach and asymptotic expansions in the vorticity--divergence representation of the fluid velocity. In the derivation of TL1, we use a modified version of the Euler--Poincar\'e framework introduced in \cite{allen1996extended} which expresses the approximate momentum in terms of gradients of advected quantities.

\item In section \ref{sec:stochEliassen} we derive the TL1 equations with stochastic advection by Lie transport (SALT) in the modified Euler--Poincar\'e variational framework. These stochastic TL1 equations would be useful for uncertainty quantification and data assimilation at this intermediate level of approximation. 

\item In section \ref{sec:TQGmodel} we take the next step in the asymptotic expansion to derive and discuss the deterministic version of the thermal quasi-geostrophic (TQG) equations in section \ref{sec:detTQG}. Section \ref{sec:TQG-numerical example} specifies the numerical details of an example implementation of the TQG solution shown in Figure \ref{Fig:CourtesyofWeiPan-Intro}. We then investigate the Hamiltonian framework of the TQG equations. The Hamiltonian formulation of the TQG equations can be used to derive the SALT stochastic TQG equations. This is described in section \ref{sec:hamTQG}.

\item In section \ref{sec:openprobs} we conclude by outlining a few next steps and open problems to which the present work has led us, but which we feel are beyond the scope of the present paper.

\end{enumerate}

\section{The thermal rotating shallow water (TRSW) model}\label{DTRSW-sec}

The thermal rotating shallow water (TRSW) model describes the motion of a single two dimensional layer of fluid with horizontally varying buoyancy and bottom topography (or, bathymetry). 
The TRSW model is an extension of the rotating shallow water model and a simplification of the various three dimensional models such as the Primitive Equations and the Euler-Boussinesq model, which are commonly used for computationally simulating large-scale ocean and atmosphere circulation dynamics. The thermal rotating shallow water equations may also be interpreted as a model for an upper active layer of fluid on top of a lower inert layer. For that reason the TRSW model is sometimes called a 1.5 layer model \cite{warneford2013quasi}. A stochastic version of this model has already been derived from a variational point of view in Appendix B of \cite{holm2019stochastic}. For a related deterministic discussion of a fully multilayer variational model with nonhydrostatic pressure, see \cite{cotter2010square}.

\subsection{Deterministic TRSW equations}
The deterministic TRSW equations in a rotating planar domain ${\cal D}\in\mathbb{R}^2$ with boundary $\partial {\cal D}$ are expressed using the following notation.  The depth is denoted $\eta =  \eta(\bx, t)$, where $\bx=(x,y)$ is the horizontal vector position, and $t$ is time.  The (nonnegative) horizontal buoyancy is written as $b(\bx,t) = \rho(\bx,t)/\bar{\rho}$, where  $\rho(\bx,t)$ is the mass density, $\bar{\rho}$ is the uniform reference mass density. The nondimensional deterministic TRSW equations for the Eulerian horizontal vector velocity $\bu(\bx,t)$, thickness $\eta(\bx, t)$, and buoyancy $b(\bx,t)$ of the active fluid layer are 
given by
\begin{equation}   
\frac{D}{Dt} \bu+ \frac{1}{{\rm Ro}}f\zh\times\bu   
= - \frac{\alpha}{{\rm Fr}^2}\nabla\big((1+\mathfrak{s}b) \zeta\big) + \frac{\mathfrak{s}}{2{\rm Fr}^2}(\alpha\zeta-h)\nabla b \, ,   
\qquad 
\frac{\partial \eta}{\partial t} + \nabla\cd (\eta\bu) = 0\, ,   
\qquad
\frac{D}{Dt}b = 0\,.
\label{trsw-eqns}  
\end{equation}
The other notation is $f$ for the Coriolis parameter, $\alpha\zeta=\eta-h$ for the free surface elevation, where $h(\bx)$ is the time-independent mean depth, $\mathfrak{s}$ for the stratification parameter, $\alpha$ for the wave amplitude, ${\rm Ro}=U/(f_0 L)$ for the Rossby number and ${\rm Fr}^2=U^2/(g H)$ for the Froude number. The material time derivative for scalar advected quantities is denoted by $\frac{D}{Dt}:=\partial_t + \bu\cdot\nabla$. In defining the dimensionless numbers, $U$ denotes the horizontal velocity scale, $L$ is the horizontal length scale, $f_0$ is the typical rotation frequency, $g$ is the gravitational acceleration and $H$ is the typical depth. The stratification parameter $\mathfrak{s}$ is introduced so that the buoyancy variable has size $\mathcal{O}(1)$ and the importance of buoyancy is governed by the size of the stratification parameter. The stratification parameter is particularly important in introducing the Boussinesq approximation. The Boussinesq approximation is necessary to derive the thermal rotating shallow water equations. The wave amplitude $\alpha$ is the typical free surface elevation $\zeta_0$ divided by the typical depth $H$ and is introduced so that $\zeta$ has size $\mathcal{O}(1)$ and the importance of free surface waves is governed by the size of the wave amplitude. The notation $\zh$ is used to denote the unit vector perpendicular to the flow domain $\cal D$. The boundary conditions are 
\begin{equation} 
\mathbf{\hat{n}}\cd \bu=0
\quad\hbox{and}\quad
\mathbf{\hat{n}}\times\nabla b = 0
\quad\hbox{on the boundary $\partial{\cal D}$,}
\label{trsw-bdy}  
\end{equation}
meaning that fluid velocity $\bu$ is tangential and buoyancy $b$ is constant on the boundary $\partial{\cal D}$, fixed in the frame rotating at time-independent angular frequency $\zh f(\bx)/2$. In the boundary conditions $\mathbf{\hat{n}}$ denotes the outward unit normal. Periodic boundary conditions may also be considered.

\paragraph{Variational formulation.} The TRSW equations \eqref{trsw-eqns} can be derived by means of the Euler--Poincar\'e variatonal principle, as is shown in \cite{brocker2018mathematics}. When the equations of motion are derived in this framework, there is a natural way to express three fundamental relations. The first fundamental relation is the Kelvin circulation theorem, the second one is the advection equation for potential vorticity and the third one is an infinity of conserved integral quantities arising from Noether's theorem for the symmetry of Eulerian fluid quantities under Lagrangian particle relabelling. For example, in rotating shallow water, without a buoyant scalar, the enstrophy is among these integral quantities. This framework turns out to be ideal for introducing stochasticity, as shown in \cite{holm2015variational, de2020implications} and one can similarly introduce rough paths \cite{crisan2020variational} into continuum mechanics. In applications, this framework provides a means to consistently introduce data-driven parametrisations of stochastic transport into a large class of fluid models \cite{cotter2018modelling,cotter2019numerically}.

\subsection{The Euler--Poincar\'e theorem}\label{sec-EPframework}

\paragraph{Variational derivatives of functionals.}
The Euler--Poincar\'e theorem relies on variational derivatives of functionals. This type of derivative is given by the following definition.
\begin{definition}\label{def:variationalderivatives}
A functional $F[\rho]$ is defined as a map  $F: \rho\in C^\infty(\cal D)\to \mathbb{R}$. The variational derivative of $F(\rho)$, denoted $\delta F/ \delta \rho$, is defined by the linear functional
\begin{align}\label{eq:var-op-def}
\delta F[\rho]
:= \lim_{\varepsilon\to 0}\frac{F[\rho+\varepsilon \phi]-F[\rho]}{\varepsilon} 
=:  \frac{d}{d\varepsilon}F[\rho+\varepsilon \phi]\bigg|_{\epsilon=0}
=
 \int_\Omega \frac{\delta F}{\delta\rho}(x) \phi(x) \; dx 
 =: \left\langle \frac{\delta F}{\delta\rho}\,,\,\phi \right\rangle
\,.\end{align}
In this definition, $\varepsilon\in \mathbb{R}$ is a real parameter, $\phi$  is an arbitrary smooth function and the angle brackets $\langle\,\cdot\,,\,\cdot\,\rangle$ indicate $L^2$ real symmetric  pairing of integrable smooth functions on the flow domain $\cal D$. The function $ \phi(x)$ above is called the `variation of $\rho$' and will be denoted as $\delta \rho:= \phi(x) $. Since the variation is a linear operator on functionals, we can define the functional derivative $\delta$ in \eqref{eq:var-op-def} operationally as
\begin{equation}
\delta F[\rho] := \left\langle \frac{\delta F}{\delta\rho}\,,\,\delta \rho \right\rangle.
\end{equation}
\end{definition}

\paragraph{Euler--Poincar\'e theorem.}
Given the boundary conditions and definitions above, the following form of the Euler--Poincar\'e theorem  will provide the deterministic equations of motion derived from Hamilton's principle. Suppose a deterministic Lagrangian functional $\ell:\mathfrak{X}\times V^*\to \mathbb{R}$ is defined on the domain of flow, $\cal D$. Here $\mathfrak{X}$ denotes the space of smooth vector fields on $\cal D$ and $V^*$ is the vector space of advected quantities. Advected quantities are tensor fields of various types which are preserved along the flow. The space of smooth vector fields is a Lie algebra under the action of the Jacobi--Lie bracket, which is denoted as $[\,\cdot\,,\,\cdot]:\mathfrak{X}\times \mathfrak{X}\to \mathfrak{X}$, and is defined for $u,v\in \mathfrak{X}$ by the commutator relation
\begin{equation}\label{eq:ad-def}
\big[u,v\big] := \big( (\mathbf{u}\cdot\nabla)\mathbf{v} 
-
(\mathbf{v}\cdot\nabla)\mathbf{u} \big)\cdot\nabla
\,.
\end{equation}
\begin{theorem}[Euler--Poincar\'e equations \cite{holm1998euler}]
\label{thm:EP}$\,$\\
The following two statements are equivalent:
\begin{enumerate}[i)]
\item Hamilton's variational principle in Eulerian coordinates, with $\mathbf{u}\in\mathfrak{X}(\cal D)$ and $b, \eta\in V^*(\cal D)$,\begin{equation}
\delta S := \delta\int_{t_1}^{t_2}\ell(\mathbf{u},b,\eta) \,dt= 0,
\label{eq:hvp}
\end{equation}
holds on $\mathfrak{X}({\cal D})\times V^*$, using variations of the form
\begin{equation}
\delta \mathbf{u} = \frac{\partial}{\partial t}\mathbf{v}- [\bu, \mathbf{v}], \qquad \delta b = -(\mathbf{v}\cdot\nabla) b \,, \qquad \delta \eta = -\nabla\cdot(\eta\mathbf{v})\,,
\end{equation}
where the vector field $\mathbf{v}\in\mathfrak{X}(\cal D)$ is arbitrary and vanishes on the endpoints $t_1$ and $t_2$.
\item The Euler--Poincar\'e equations hold. These equations are
\begin{equation}
\frac{\partial}{\partial t}\frac{\delta \ell}{\delta \mathbf{u}} 
+ (\bu\cdot\nabla)\frac{\delta\ell}{\delta \mathbf{u}} 
+ (\nabla\bu)\cdot\frac{\delta\ell}{\delta \mathbf{u}} 
+ \frac{\delta\ell}{\delta \mathbf{u}} (\nabla\cdot \bu)
= -\frac{\delta \ell}{\delta b}\nabla b + \eta\nabla\frac{\delta\ell}{\delta \eta}
\,,\label{eq:EPeq}
\end{equation}
or, equivalently, in two-dimensional vector calculus notation,
\begin{equation}
\frac{\partial}{\partial t} \frac{\delta \ell}{\delta \mathbf{u}} 
+ \left(\nabla^\perp \cdot  \frac{\delta\ell}{\delta \mathbf{u}} \right)\bu^\perp  
+ \nabla \left(\bu \cdot\frac{\delta\ell}{\delta \mathbf{u}} \right)
+ \frac{\delta\ell}{\delta \mathbf{u}} (\nabla\cdot\bu)
= -\frac{\delta \ell}{\delta b}\nabla b + \eta\nabla\frac{\delta\ell}{\delta \eta}
\,,\label{eq:EPeq-Cartan2D}
\end{equation}
or, finally, as an embedding in three dimensional space,
\begin{equation}
\frac{\partial}{\partial t} \frac{\delta \ell}{\delta \mathbf{u}} 
- \bu\times \left(\nabla \times  \frac{\delta\ell}{\delta \mathbf{u}} \right) 
+ \nabla \left(\bu \cdot\frac{\delta\ell}{\delta \mathbf{u}} \right)
+ \frac{\delta\ell}{\delta \mathbf{u}} (\nabla\cdot\bu)
= -\frac{\delta \ell}{\delta b}\nabla b + \eta\nabla\frac{\delta\ell}{\delta \eta}
\,,\label{eq:EPeq-Cartan}
\end{equation}
with advection equations 
\begin{align}
\frac{\partial}{\partial t} b = -\, \bu \cdot \nabla b
\quad\hbox{and}\quad
\frac{\partial}{\partial t} \eta = -\, \nabla\cdot(\eta\bu)
\,.
\label{eq:EB-advectionLaws}
\end{align}
\end{enumerate}
\end{theorem}

\begin{remark}
The abstract statement of the Euler--Poincar\'e Theorem \ref{thm:EP}, formulated on general semidirect product Lie groups, is presented in \cite{holm1998euler} deterministically, in \cite{holm2015variational,de2020implications} stochastically and finally in \cite{crisan2020variational} on rough paths.
\end{remark}\medskip

\begin{remark}\label{rem:perpnotation}
In Theorem \ref{thm:EP}, the operator $\delta$ in \eqref{eq:hvp} is the functional derivative defined in \eqref{eq:var-op-def}, the brackets $[\,\cdot\,,\,\cdot\,]$ denote the commutator of vector fields defined in \eqref{eq:ad-def}, and $\mathbf{v} \in\mathfrak{X}(\cal D)$ is an arbitrary vector field in two dimensions which vanishes at the endpoints in time, $t_1$ and $t_2$. Equations \eqref{eq:EPeq-Cartan2D} and \eqref{eq:EPeq-Cartan} are equivalent and can be transformed into each other by the conventions $\bu^\perp = \zh\times\bu$ and $\nabla^\perp = \zh\times\nabla$, where $\zh$ is the outward unit vector perpendicular to the planar domain $\cal D$.
\end{remark}\medskip

\begin{remark}
One may interpret the Euler--Poincar\'e equation \eqref {eq:EPeq} in terms of Newton's law for fluid motion. Namely, the rate of change of the covector momentum density $\mathbf{P}:=\delta\ell/\delta \mathbf{u}$ along the flow of the fluid velocity vector field $u$ equals the sum of covector force densities on the right hand side of equation \eqref{eq:EPeq}.
\end{remark}

\begin{proof}
Hamilton's variational principle implies
\begin{align*}
0 &= \int_{t_1}^{t_2}\left[ \left\langle\frac{\delta\ell}{\delta\mathbf{u}},\delta\mathbf{u} \right\rangle_\mathfrak{X} + \left\langle \frac{\delta\ell}{\delta b},\delta b \right\rangle_{V^*} + \left\langle\frac{\delta\ell}{\delta \eta},\delta \eta \right\rangle_{V^*}\right]\,dt
\\
&=  \int_{t_1}^{t_2}\left[ \left\langle\frac{\delta\ell}{\delta\mathbf{u}},\frac{\partial}{\partial t}\mathbf{v} - [\bu, \mathbf{v}]\right\rangle_\mathfrak{X} + \left\langle \frac{\delta\ell}{\delta b},-(\mathbf{v}\cdot\nabla)b \right\rangle_{V^*} + \left\langle\frac{\delta\ell}{\delta \eta},-\nabla\cdot(\eta \mathbf{v})\right\rangle_{V^*}\right]\, dt
\\
&= \int_{t_1}^{t_2}\left[ \left\langle -\frac{\partial}{\partial t} \frac{\delta\ell}{\delta \mathbf{u}}- (\bu \cdot\nabla)\frac{\delta\ell}{\delta \mathbf{u}} - (\nabla\bu)\cdot\frac{\delta\ell}{\delta \mathbf{u}} + \frac{\delta\ell}{\delta\mathbf{u}}(\nabla\cdot\bu),\mathbf{v}\right\rangle_\mathfrak{X} + \left\langle -\frac{\delta\ell}{\delta b}\nabla b, \mathbf{v}\right\rangle_\mathfrak{X}\right.\\
&\qquad\quad \left.+ \left\langle \eta\nabla\frac{\delta\ell}{\delta \eta}, \mathbf{v}\right\rangle_\mathfrak{X} \right]\,dt.
\end{align*}
The subscripts $\mathfrak{X}$ and $V^*$ on the $L^2$ pairings indicate over which space that the pairing is defined. Since $\mathbf{v}$ is arbitrary and vanishes at the endpoints $t_1$ and $t_2$ in time, the following equation holds,
\begin{align*}
\frac{\partial}{\partial t}\frac{\delta \ell}{\delta \mathbf{u}} + (\bu \cdot\nabla)\frac{\delta\ell}{\delta \mathbf{u}} + (\nabla\bu)\cdot\frac{\delta\ell}{\delta \mathbf{u}} + \frac{\delta\ell}{\delta\mathbf{u}}(\nabla\cdot\bu)= -\frac{\delta \ell}{\delta b}\nabla b + \eta\nabla\frac{\delta\ell}{\delta \eta}\,.
\end{align*}
This finishes the proof of the stochastic Euler--Poincar\'e equation in \eqref{eq:EPeq}. The equivalent forms in equations \eqref{eq:EPeq-Cartan2D} and \eqref{eq:EPeq-Cartan} follow by means of a standard vector identity. 
\end{proof}

\subsection{Kelvin--Noether circulation theorem}\label{sec-KNthm}
A straightforward calculation using the second advection equation in \eqref{eq:EB-advectionLaws} shows that \eqref{eq:EPeq} may be written equivalently as follows.
\medskip

\begin{lemma}\label{lemma:circulationform}
The Euler--Poincar\'e equation in \eqref{eq:EPeq} is equivalent to the following,
\begin{equation}
\frac{\partial}{\partial t}\left(\frac{1}{\eta}\frac{\delta \ell}{\delta \mathbf{u}}\right)
+ (\bu\cdot\nabla)\left(\frac{1}{\eta}\frac{\delta \ell}{\delta \mathbf{u}}\right)
+ (\nabla \bu)\cdot\left(\frac{1}{\eta}\frac{\delta \ell}{\delta \mathbf{u}}\right)
= -\,\frac{1}{\eta}\frac{\delta \ell}{\delta b}\nabla b + \nabla\frac{\delta\ell}{\delta \eta}
\,.\label{eq:EPeq-circ}
\end{equation}
\end{lemma}
\medskip

\noindent One of the main features of Theorem \ref{thm:EP} for fluid dynamics is that its Euler--Poincar\'e equations satisfy the following Kelvin circulation theorem.
\medskip

\begin{theorem}[Kelvin-Noether circulation] \label{thm:Kelvin}
For an arbitrary loop $c(t)$ which is advected by the velocity field $\bu$, the following dynamics holds for the circulation integral ${\mathcal{I}}$, given by
\begin{equation}\label{eq:EPeq-circ-thm}
{\mathcal{I}} := \oint_{c(\bu)}\frac{1}{\eta}\frac{\delta\ell}{\delta \mathbf{u}}\cdot d\mathbf{x} 
\,,\qquad
\frac{d\mathcal{I}}{d t} 
= -\oint_{c(\bu)} \left(\frac{1}{\eta}\frac{\delta\ell}{\delta b}\right)\nabla b\cdot d\mathbf{x}\,.
\end{equation}
\end{theorem}
\begin{remark}\label{rem:kelvin}
The notation $c(\bu)$ indicates that the material loop $c$ is transported by the flow $\phi_t$ which is generated by the vector field $u:=\bu\cdot\nabla$. To be precise, $c(\bu) = \phi_{t*}c(0)$, where $\phi_{t*}$ is the pull-back by the inverse of the flow $\phi_t$, also known as the push-forward by $\phi_t$.
\end{remark}
\begin{proof}
The Kelvin circulation law \eqref{eq:EPeq-circ-thm} follows from Newton's law of motion obtained from the  Euler--Poincar\'e equation \eqref {eq:EPeq-circ} for the evolution of momentum per unit mass $\eta^{-1}\delta\ell/\delta \mathbf{u}$ concentrated on an advecting material loop, $c(t) = \phi_t c(0)$, where $\phi_t$ is the flow map which is generated by the vector field $\bu$. Upon changing variables by pulling back the integrand to its initial position, the time derivative can be moved inside the integral and the product rule may be applied. Then, by inverting the pull-back we obtain the following
\begin{equation}
\begin{aligned}
 \frac{d}{d t}\oint_{c(\bu)}\frac{1}{\eta}\frac{\delta \ell}{\delta \mathbf{u}}\cdot d\mathbf{x} 
&= \oint_{c(\bu)} \left(\frac{\partial}{\partial t} + \bu\cdot\nabla 
+ (\nabla \bu)\cdot\right)\left(\frac{1}{\eta}\frac{\delta \ell}{\delta \mathbf{u}}\right)\cdot d\mathbf{x}
\\
&= 
-\oint_{c(\bu)} \frac{1}{\eta}\frac{\delta \ell}{\delta b}\nabla b \cdot d\mathbf{x} 
+ \oint_{c(\bu)}\nabla \frac{\delta\ell}{\delta \eta} \cdot d\mathbf{x}
\\
&=  
-\oint_{c(\bu)} \left(\frac{1}{\eta}\frac{\delta \ell}{\delta b}\right)\nabla b \cdot d\mathbf{x}\,.
\end{aligned}
\label{proof:Kelvin-det}
\end{equation}
In the second line, we have used the Euler--Poincar\'e equation \eqref{eq:EPeq} and the advection equation for the density. The last step applies the fundamental theorem of calculus to prove vanishing of the last loop integral in the second line. For the corresponding proof in the general case, see \cite{holm1998euler}.
\end{proof}
\medskip

\begin{corollary}\label{KN-Stokes} 
By Stokes Law, equation  \eqref{eq:circLaw} in the Kelvin--Noether circulation theorem \ref{thm:Kelvin} implies 
\begin{equation}\label{eq:circLaw}
\frac{d\mathcal{I}}{d t}  = 
-\int\!\!\int_{\partial S = c(\bu)} 
\zh\cdot\nabla\left(\frac{1}{\eta}\frac{\delta\ell}{\delta b}\right)\times\nabla b\, dx\,dy\,.
\end{equation}
Therefore, circulation is created by misalignment of the gradients of buoyancy $b$ and its dual quantity $\eta^{-1}\delta\ell/\delta b$.
\end{corollary}
\medskip

\noindent The thermal rotating shallow water equations \eqref{trsw-eqns} in a planar domain $\cal{D}$ are obtained by applying the Euler--Poincar\'e theorem \ref{thm:EP} to the Lagrangian
\begin{equation}
\begin{aligned}
\ell_{TRSW}(\mathbf{u},\eta,b)
&=
\int_{\cal D} \frac{1}{2} \eta|\bu|^2 + \frac{1}{{\rm Ro}}\eta\bu\cdot\bR(\bx) - \frac{1}{2\,{\rm Fr}^2} (1+\mathfrak{s}b)(\eta^2-2\eta h) \,dx\,dy\,.
%\\
%&=
%\int_{\cal D} \left(\frac{1}{2} |\bu|^2 + \frac{1}{{\rm Ro}}\bu\cdot\bR(\bx) - \frac{1}{2\,{\rm Fr}^2} (1+\mathfrak{s}b)(\eta-2h)\right) \eta\,dx\,dy\,.
\end{aligned}
\label{lag:TRSW}
\end{equation}
This is easily shown once we have computed the variational derivatives of the Lagrangian, since these derivatives can simply be substituted into \eqref{eq:EPeq}, or into one of the equivalent formulations \eqref{eq:EPeq-Cartan2D} or \eqref{eq:EPeq-Cartan}. These variational derivatives are obtained by using definition \ref{def:variationalderivatives} to find,
\begin{equation}
\begin{aligned}
\frac{\delta\ell_{TRSW}}{\delta\bu}  &= \eta\bu + \eta \bR,\\
\frac{\delta\ell_{TRSW}}{\delta\eta} &= \frac{1}{2}|\bu|^2 + \frac{1}{{\rm Ro}}\bu\cdot \bR - \frac{\alpha}{{\rm Fr}^2}(1+\mathfrak{s}b)\zeta,\\
\frac{\delta\ell_{TRSW}}{\delta b} &= -\frac{\mathfrak{s}}{2\,{\rm Fr}^2}(\eta^2 - 2\eta h).
\end{aligned}
\label{eq:trsw-varderivatives}
\end{equation}
With these variations, we obtain \eqref{trsw-eqns} and by means of theorem \ref{thm:Kelvin}, we see that the TRSW equations satisfy the following Kelvin circulation theorem.
\medskip

\begin{theorem}[Kelvin theorem for deterministic TRSW]\label{Kel-thmTRSW}
The deterministic TRSW equations \eqref{trsw-eqns} imply the following Kelvin circulation law
\begin{equation}
\frac{d}{dt}\oint_{c(u)}\left(\bu + \frac{1}{{\rm Ro}}\bR\right) \cdot d\bx 
= \frac{\mathfrak{s}}{2\,{\rm Fr}^2}  \oint_{c(u)} (\alpha\zeta-h) \nabla b \cdot d\bx
= \frac{\mathfrak{s}}{2\,{\rm Fr}^2}  \int\!\!\int_{\partial S = c(u)}
\zh\cdot\nabla(\alpha\zeta-h) \times\nabla b\, dx\,dy
\,,
\label{Kel-trsw}   
\end{equation}
where $c(u)$ is any closed loop moving with horizontal fluid velocity $\bu(\bx,t)$ in two dimensions. 
\end{theorem}\medskip

\begin{proof}
This result follows from the Kelvin--Noether theorem  \ref{thm:Kelvin} for Euler--Poincar\'e fluid equations. 
\end{proof}

\begin{remark}\label{rem:keltrsw}
Equation \eqref{Kel-trsw} implies that misalignment between the horizontal gradients of either the free surface elevation $\alpha\zeta$, or the bathymetry $h$ with the buoyancy $b$ will generate circulation in a horizontal plane. When the free surface elevation is negligible, that is when the wave amplitude $\alpha\ll 1$, circulation is still being generated due to gradients of bathymetry misaligning with gradients of buoyancy. When the contribution of buoyancy is negligible, that is when the stratification parameter $\mathfrak{s}\ll 1$, circulation is conserved. In the sections to come, we will use theorem \ref{Kel-trsw} to interpret the properties of the approximate equations we will derive.
\end{remark}\medskip

\begin{corollary}[Circulation on the boundary]\label{Kel-corollary}
The circulation $\oint_{\,\Gamma_k}\left(\bu + \frac{1}{{\rm Ro}}\bR\right) \cdot d\bx$ on each connected component of the boundary $\Gamma_k\in \partial\cal{D}$ is conserved by the deterministic TRSW equations.
\end{corollary}

\begin{proof}
Preservation of circulation on each connected component of the boundary follows from the boundary conditions in \eqref{trsw-bdy} and Kelvin's theorem for TRSW in \eqref{Kel-trsw}. The first boundary condition in \eqref{trsw-bdy} implies that the velocity is tangent to the boundary. Hence, a circuit $c(\bu)$ on the boundary remains on the boundary. Consequently, Kelvin's theorem for TRSW in \eqref{Kel-trsw} applies to a boundary circuit.
The second boundary condition in \eqref{trsw-bdy} implies that $\nabla b \cdot d\bx = 0$ on the boundary $\partial{\cal D}$. Hence, the right-hand side of \eqref{Kel-trsw} vanishes for a circuit $c(\bu)$ on the boundary and the circulation $\oint_{\,\Gamma_k}\left(\bu + \frac{1}{{\rm Ro}}\bR\right) \cdot d\bx$ is conserved on the boundary.  
\end{proof}
\medskip

\noindent  
The potential vorticity for the thermal rotating shallow water equations \eqref{trsw-eqns} is defined as
\begin{equation}
q := \frac1\eta \zh\cdot\nabla\times\Big(\bu + \frac{1}{{\rm Ro}}\bR\Big) \,.
\label{def:pvtrsw}
\end{equation}
Even though this is the same definition of potential vorticity as for the rotating shallow water equations, in thermal rotating shallow water, the potential vorticity is not conserved along Lagrangian paths. Rather, the potential vorticity satisfies the following equation
\begin{equation}
\frac{\partial}{\partial t}q + \bu\cdot\nabla q = \frac{\mathfrak{s}}{2\,{\rm Fr}^2\eta}\zh\cdot \nabla(\alpha\zeta-h)\times\nabla b\,.
\label{eq:pvtrsw}
\end{equation}
Not unexpectedly, the mechanism responsible for the generation of circulation in the Kelvin circulation theorem \ref{Kel-thmTRSW} is also rate of creation of potential vorticity, $q$, along fluid particle trajectories in equation \eqref{eq:pvtrsw}. When the horizontal buoyancy gradients are negligible, that is $\mathfrak{s}\ll 1$, potential vorticity is conserved along Lagrangian paths.

\paragraph{Conservation laws for deterministic TRSW.} The deterministic TRSW equations \eqref{trsw-eqns} conserve the energy
\begin{equation}
\mathcal E_{TRSW}(\bu,\eta,b) =  \frac12 \int_{\cal D} \eta |\bu|^2 + \frac{1}{{\rm Fr}^2}(1+\mathfrak{s}b)(\eta^2-2\eta h)\,dx\,dy\,.
\label{eq:erg-trsw}
\end{equation}
The conservation of energy \eqref{eq:erg-trsw} can be proved directly by using the TRSW equations \eqref{trsw-eqns} and the boundary conditions \eqref{trsw-bdy}. The TRSW equations \eqref{trsw-eqns} also conserve an infinity of integral conservation laws, determined by two arbitrary differentiable functions of buoyancy $\Phi(b)$ and $\Psi(b)$ as
\begin{equation}
C_{\Phi,\Psi} = \int_{\cal D} \eta\Phi(b) + \varpi\Psi(b) \,dx\,dy = \int_{\cal D} \Big( \Phi(b) + q \Psi(b)\Big)\eta\,dx\,dy\,,
\label{eq:casimirstrsw}
\end{equation}
where $\varpi = \eta q= \zh\cdot\nabla\times(\bu+\frac{1}{{\rm Ro}}\bR)$ is the total vorticity. That is, for any choice of differentiable $\Phi$ and $\Psi$, the quantity $C_{\Phi,\Psi}$ is conserved in time. The conservation of \eqref{eq:casimirsstrsw} can also be proved as a direct calculation using equations \eqref{trsw-eqns} and the boundary conditions \eqref{trsw-bdy}. 

\paragraph{Noether's theorem.} Conservation of the integral quantities in equations \eqref{eq:erg-trsw} and \eqref{eq:casimirstrsw} is associated by Noether's theorem with smooth transformations which leave invariant the Eulerian fluid quantities in the Lagrangian \cite{abraham1978foundations}. For example, conservation of energy \eqref{eq:erg-trsw} arises from invariance of the Lagrangian in \eqref{lag:TRSW} under translations in time; since this Lagrangian does not depend explicitly on time. Likewise, the conserved quantities in \eqref{eq:casimirstrsw} are associated by Noether's theorem with the smooth flows which translate the fluid parcels along steady solutions of the equations of motion; since, of course these transformations preserve the Eulerian fluid variables in the Lagrangian \cite{holm1985nonlinear}. Upon introducing stochasticity via the Euler--Poincar\'e theorem in the next section, the latter transformations and their Noether conservation laws will persist. However, energy conservation will not persist because the stochastic Lagrangian will depend explicitly on time through the Brownian noise. The geometrical significance of the conservation laws in equation \eqref{eq:casimirstrsw} which persist for stochastic TRSW will be discussed further in remark \ref{remark:Casimirs}. 

\subsection{TRSW with stochastic advection by Lie transport (SALT)}\label{STRSW-sec}
By modifying the fluid transport vector field in the Euler--Poincar\'e theorem \ref{thm:EP}, one can derive the stochastic equations of motion which preserve the geometric properties of their deterministic counterparts. Following \cite{holm2015variational}, we introduce the stochastic vector field for fluid transport  in semimartingale form
\begin{equation}
{\sf d}\boldsymbol \chi_{t} := \mathbf{u}(\mathbf{x},t)dt + \sum_{k=1}^M \boldsymbol \xi_k(\mathbf{x})\circ dW_t^k\,,
\label{def:chi}
\end{equation}
where $\bx\in{\cal D}\subset\mathbb{R}^2$. The time-independent vector fields $\boldsymbol \xi_k$, with $k=1,2,\dots,M$, represent spatially correlated sources of temporal uncertainty. The vector fields $\boldsymbol \xi_k$ must satisfy the same boundary condition as $\bu$. The circle notation $\circ$ means that the stochastic integral is to be understood in the Stratonovich sense. Stratonovich calculus possesses the ordinary chain rule and product rule, which are crucial for defining the variational derivative. These properties are written in integral form, though, because stochastic equations are not differentiable with respect to time. The sources of the stochasticity are the independent, identically distributed Brownian motions $W_t^k$ associated to each $\boldsymbol \xi_k$. The Brownian motions are defined with respect to the standard probability space, see \cite{ito_1984}.
%\url{https://en.wikipedia.org/wiki/Standard_probability_space}. 
One may regard the $\boldsymbol \xi_k(\mathbf{x})$ as eigenvectors of the velocity-velocity correlation tensor. In practice, the eigenvectors $\boldsymbol \xi_k(\mathbf{x})$ are expected to be obtained using the SALT algorithm developed in \cite{cotter2018modelling,cotter2019numerically}. This algorithm is based on empirical orthogonal function analysis, and the number $M$ of $\boldsymbol \xi_k$ needed in \eqref{def:chi} would be decided by how much of the variance is required to be represented. \medskip

\noindent
The SALT version of Theorem \ref{thm:EP} may be stated, as follows. 
\begin{theorem}[Stochastic Euler--Poincar\'e equations \cite{holm2015variational,de2020implications}]
\label{thm:SEP}$\,$\\
The following two statements are equivalent: 
\begin{enumerate}[i)]
\item The stochastic Hamilton's variational principle in Eulerian coordinates, with $\mathbf{u}\in\mathfrak{X}(\cal D)$ and $b, \eta\in V^*(\cal D)$,
\begin{equation}
\delta S := \delta\int_{t_1}^{t_2}\ell(\mathbf{u},b,\eta) \,dt= 0,
\label{eq:shvp}
\end{equation}
holds on $\mathfrak{X}({\cal D})\times V^*$, using variations of the form
\begin{equation}
\delta \mathbf{u} \,dt = {\sf d}\mathbf{v} - [{\sf d}\boldsymbol \chi_t, \mathbf{v}], \qquad \delta b\, dt = -(\mathbf{v}\cdot\nabla) b \,dt \,, \qquad \delta \eta\, dt = -\nabla\cdot(\eta\mathbf{v})\, dt\,,
\end{equation}
where the vector field $\mathbf{v}\in\mathfrak{X}(\cal D)$ is arbitrary and vanishes on the endpoints $t_1$ and $t_2$ and the semimartingale vector field ${\sf d}\boldsymbol \chi_t$ is defined in \eqref{def:chi}.
\item The stochastic Euler--Poincar\'e equations hold. These equations are
\begin{equation}
{\sf d}\frac{\delta \ell}{\delta \mathbf{u}} 
+ ({\sf d}\boldsymbol \chi_t\cdot\nabla)\frac{\delta\ell}{\delta \mathbf{u}} 
+ (\nabla {\sf d}\boldsymbol \chi_t)\cdot\frac{\delta\ell}{\delta \mathbf{u}} 
+ \frac{\delta\ell}{\delta \mathbf{u}} (\nabla\cdot {\sf d}\boldsymbol \chi_t)
= -\frac{\delta \ell}{\delta b}\nabla b\,dt + \eta\nabla\frac{\delta\ell}{\delta \eta}\,dt
\label{eq:SEPeq}
\end{equation}
or, equivalently, either in two dimensional vector calculus notation,
\begin{equation}
{\sf d} \frac{\delta \ell}{\delta \mathbf{u}} 
+ \left(\nabla^\perp \cdot  \frac{\delta\ell}{\delta \mathbf{u}} \right) {\sf d}\boldsymbol \chi_t^\perp
+ \nabla \left({\sf d}\boldsymbol \chi_t \cdot\frac{\delta\ell}{\delta \mathbf{u}} \right)
+ \frac{\delta\ell}{\delta \mathbf{u}} (\nabla\cdot{\sf d}\boldsymbol \chi_t)
= -\frac{\delta \ell}{\delta b}\nabla b\,dt + \eta\nabla\frac{\delta\ell}{\delta \eta}\,dt
\,,\label{eq:SEPeq-Cartan2D}
\end{equation}
or as an embedding in three dimensional space,
\begin{equation}
{\sf d} \frac{\delta \ell}{\delta \mathbf{u}} 
- {\sf d}\boldsymbol \chi_t \times \left(\nabla \times  \frac{\delta\ell}{\delta \mathbf{u}} \right) 
+ \nabla \left({\sf d}\boldsymbol \chi_t \cdot\frac{\delta\ell}{\delta \mathbf{u}} \right)
+ \frac{\delta\ell}{\delta \mathbf{u}} (\nabla\cdot{\sf d}\boldsymbol \chi_t)
= -\frac{\delta \ell}{\delta b}\nabla b\,dt + \eta\nabla\frac{\delta\ell}{\delta \eta}\,dt
\,,\label{eq:SEPeq-Cartan}
\end{equation}
with advection equations 
\begin{align}
{\sf d} b = -\, {\sf d}\boldsymbol \chi_t \cdot \nabla b
\quad\hbox{and}\quad
{\sf d} \eta = -\, \nabla\cdot(\eta\,{\sf d}\boldsymbol \chi_t)
\,.
\label{eq:SEB-advectionLaws}
\end{align}
\end{enumerate}
\end{theorem}
\noindent For the proof of this theorem and the technical details we refer to \cite{holm2015variational,de2020implications}. By taking the variational derivatives of the Lagrangian for thermal rotating shallow water as in \eqref{eq:trsw-varderivatives}, we obtain the stochastic TRSW equations
\begin{equation}
\begin{aligned}
{\sf d}\bu + ({\sf d}\boldsymbol \chi_t\cdot\nabla)\bu  + \sum_k(\nabla \boldsymbol \xi_k)\cdot\bu\circ dW_t^k 
&= - \frac{\alpha}{{\rm Fr}^2}\nabla\big((1+\mathfrak{s}b) \zeta\big)\,dt + \frac{\mathfrak{s}}{2\,{\rm Fr}^2}(\alpha\zeta-h)\nabla b\,dt 
\\&\qquad 
- \frac{1}{{\rm Ro}}f\hat{\mathbf{z}}\times {\sf d}\boldsymbol \chi_t - \frac{1}{{\rm Ro}}\sum_k\nabla(\boldsymbol \xi_k\cdot\mathbf{R})\circ dW_t^k
\,,\\
{\sf d}\eta + \nabla\cdot(\eta\, {\sf d}\boldsymbol \chi_t) &= 0
\,,\\
{\sf d} b + ({\sf d}\boldsymbol \chi_t\cdot\nabla)b &= 0\,.
\end{aligned}
\label{eq:STRSW}
\end{equation}
The boundary conditions are given by
\begin{equation}
\mathbf{\hat{n}}\cd \bu=0
\quad\hbox{and}\quad
\mathbf{\hat{n}}\cd \boldsymbol{\xi}_k = 0
\quad\hbox{and}\quad
\mathbf{\hat{n}}\times\nabla b = 0
\quad\hbox{on the boundary $\partial{\cal D}$.}
\label{eq:bdySTRW}
\end{equation}
The boundary condition on $\boldsymbol \xi_k$ is required to be satisfied for each $k$. The Kelvin circulation theorem has now become stochastic, because the circulation loop is transported by the stochastic vector field ${\sf d}\boldsymbol \chi_t$, rather than by the deterministic vector field $\bu$. Specifically, we have:
\begin{theorem}\label{KelThm-stochTRSW}
The stochastic Kelvin circulation law associated to the stochastic Euler--Poincar\'e theorem is
\begin{equation}
{\sf d}\oint_{c({\sf d}\boldsymbol \chi_t)} \frac{1}{\eta} \frac{{\delta} \ell}{{\delta} \mathbf{u}} \cdot d\bx 
= -\, \oint_{c({\sf d}\boldsymbol \chi_t)} \frac{1}{\eta} \frac{{\delta} \ell}{{\delta} b} \nabla b \cdot d\bx \,dt
\,,
\label{EP-Kelvin-Noether-trsw}   
\end{equation}
where $c({\sf d}\boldsymbol \chi_t)$ is a closed loop that is transported by the flow generated by the stochastic fluid velocity ${\sf d}\boldsymbol \chi_t$ in two dimensions.
\end{theorem}
\begin{proof}
By following the proof \eqref{proof:Kelvin-det} of the deterministic Kelvin circulation theorem \ref{thm:Kelvin} and using the product rule and chain rule for the stochastic time differential {\sf d}, we have
\begin{equation}
\begin{aligned}
{\sf d}\oint_{c({\sf d}\boldsymbol \chi_t)}\frac{1}{\eta}\frac{\delta \ell}{\delta \mathbf{u}}\cdot d\mathbf{x} 
&= \oint_{c({\sf d}\boldsymbol \chi_t)} \left({\sf d} + {\sf d}\boldsymbol \chi_t\cdot\nabla 
+ (\nabla {\sf d}\boldsymbol \chi_t)\,\cdot\right)\left(\frac{1}{\eta}\frac{\delta \ell}{\delta \mathbf{u}}\right)\cdot d\mathbf{x}\, dt
\\
&= 
-\oint_{c({\sf d}\boldsymbol \chi_t)} \frac{1}{\eta}\frac{\delta \ell}{\delta b}\nabla b \cdot d\mathbf{x}\, dt
+ \oint_{c({\sf d}\boldsymbol \chi_t)}\nabla \frac{\delta\ell}{\delta \eta} \cdot d\mathbf{x}\, dt
\\
&=  
-\oint_{c({\sf d}\boldsymbol \chi_t)} \left(\frac{1}{\eta}\frac{\delta \ell}{\delta b}\right)\nabla b \cdot d\mathbf{x} dt\,.
\end{aligned}
\end{equation}
\end{proof}
\begin{remark}
For the stochastic TRSW equations \eqref{eq:STRSW}, we have
\begin{equation}
{\sf d}\oint_{c({\sf d}\boldsymbol \chi_t)}\left(\bu+\frac{1}{{\rm Ro}}\bR\right) \cdot d\bx 
= \frac{\mathfrak{s}}{2\,{\rm Fr}^2}  \oint_{c({\sf d}\boldsymbol \chi_t)} (\alpha\zeta-h) \nabla b \cdot d\bx\,dt
= \frac{\mathfrak{s}}{2\,{\rm Fr}^2}  \int\!\!\int_{\partial S = c({\sf d}\boldsymbol{\chi_t})}
\zh\cdot\nabla(\alpha\zeta-h) \times\nabla b \,dx\,dy\,dt
\,.
\label{EP-KN-trsw}   
\end{equation}
One sees in equation \eqref{EP-KN-trsw} that misalignment of the horizontal gradient of either the free surface elevation $\zeta$, or the bathymetry $h$ with the horizontal gradient of the buoyancy $b$ will generate circulation, cf. the corresponding deterministic TRSW Kelvin circulation theorem in equation \eqref{Kel-trsw}. 
\end{remark}
\bigskip

\begin{remark}
The evolution of potential vorticity on fluid parcels for the TRSW equations in \eqref{eq:STRSW} is given by
\begin{equation}
{\sf d}q + ({\sf d}\boldsymbol \chi_t\cdot\nabla)q = \frac{\mathfrak{s}}{2\,{\rm Fr}^2\,\eta}\zh\cdot\nabla(\alpha\zeta-h)\times\nabla b\,dt,
\end{equation}
where the potential vorticity is defined by
\begin{equation}
q:= \frac{\varpi}{\eta}, \qquad \hbox{and} \qquad \varpi := \hat{\mathbf{z}}\cdot\nabla\times\left(\bu+\frac{1}{{\rm Ro}}\bR\right).
\end{equation}
\end{remark}
\bigskip

\begin{remark}
The stochastic TRSW equations \eqref{eq:STRSW} have an infinite number of conserved integral quantities
\begin{equation}
C_{\Phi,\Psi} = \int_{\cal D} \Big( \Phi(b) +  q\, \Psi(b)\Big) \eta\,dxdy,
\label{eq:casimirsstrsw}
\end{equation}
for the boundary conditions given in \eqref{trsw-bdy} and any differentiable functions $\Phi$ and $\Psi$.
\end{remark}
\bigskip

\begin{remark}[Stochastic Hamiltonian formulation]\label{remark:Casimirs}
The Legendre transform which determines the Hamiltonian ${\sf d}h_{TRSW}$ for the stochastic TRSW equations is defined as
\begin{equation}
{\sf d}\hslash_{TRSW}  (\mu ,\eta,b)
:= 
\big\langle \mu , {\sf d}\boldsymbol \chi_t \big\rangle - \ell_{TRSW} (\mathbf{u},\eta,b)dt
\,,
\label{LegXform-VCH092}
\end{equation}
in which the angle brackets in the definition of the Legendre transform denote the $L^2$ pairing over the domain ${\cal D}$. Notice that the Hamiltonian ${\sf d}h_{TRSW}$ in \eqref{LegXform-VCH092} is a semimartingale. See \cite{street2020semi} for variational principles driven by semimartingales. The Hamiltonian form of the stochastic TRSW equations is given for a functional $F(\mu_i,\eta,b)$ by
\begin{align}
{\sf d} F =\Big\{ F, {\sf d}\hslash_{TRSW}\Big\}
=
- \int_{\cal D}
\begin{bmatrix}
{\delta F}/{\delta \mu_i} \\  {\delta F}/{ \delta \eta} \\ {\delta F}/{ \delta b} 
\end{bmatrix}^T
\begin{bmatrix}
\mu_j\partial_i + \partial_j \mu_i & \eta \partial_i & -\,b_{,i}
\\ 
\partial_j \eta & 0 & 0 \\
b_{,j} & 0 & 0
\end{bmatrix}
\begin{bmatrix}
{\delta ({\sf d}\hslash_{TRSW} }) / {\delta \mu_j}\
\\
{\delta ({\sf d}\hslash_{TRSW} })/{ \delta \eta}
\\  
{\delta ({\sf d}\hslash_{TRSW} })/{ \delta b}
\end{bmatrix}
\,dx\,dy\,.
\label{Ham-matrix-EB-stoch}
\end{align}
In this notation, repeated indices are summed over. The conserved integral quantities $C_{\Phi,\Psi}$ defined in \eqref{eq:casimirsstrsw} are Casimirs of the Lie--Poisson bracket in \eqref{Ham-matrix-EB-stoch}. That is, the vector of variational derivatives of $C_{\Phi,\Psi}$ comprises a null eigenvector of the Lie--Poisson bracket in \eqref{Ham-matrix-EB-stoch}. Consequently, their conservation persists when the Hamiltonian is made stochastic.  This means that the solutions of these equations describe stochastic coadjoint motion in function space on level sets of the Casimir functionals $C_{\Phi,\Psi}$. Thus, the introduction of SALT into the TRSW equations preserves the Lie--Poisson bracket in their Hamiltonian formulation and thereby preserves their geometric interpretation as coadjoint motion \cite{holm2009geometric}. 
\end{remark}

\section{Balanced interpretations of TRSW} \label{sec:detEliassen}
There exist several approximations of the rotating shallow water (RSW) equations, the most famous one being the quasi-geostrophic (QG) approximation. By assuming the motion to take place in a particular scaling regime, it can be shown that the largest component of the velocity field, called the geostrophic velocity field, is determined by a diagnostic equation, rather than a prognostic equation. The QG approximation is a small perturbation around this geostrophic velocity field. There exists an intermediate model which is more accurate than QG, but is still an approximation of RSW. In this section, we will derive the thermal geostrophic balance by identifying the correct scaling regime and use asymptotic expansions to simplify the TRSW equations. Next, we will show that the thermal rotating shallow water equations can be approximated geometrically to derive a class of equations which was first proposed by Eliassen \cite{eliassen1949quasi} and made into a variational theory by \cite{salmon1983practical}, where it is called L1. The Lagrangian corresponding to the equations proposed by Eliassen can be obtained via two approaches. The first approach involves the Helmholtz decomposition and the second approach follows \cite{allen1996extended}. The methods of \cite{allen1996extended} will be applied in the Euler--Poincar\'e framework to derive the corresponding equations of motion. Finally, the stochastic  thermal L1 (TL1) equations will be derived via the stochastic Euler--Poincar\'e theorem. 

\subsection{Thermal geostrophic balance}
To obtain the thermal geostrophic balance relation, we return to the nondimensional deterministic TRSW equations for the 
Eulerian horizontal vector velocity $\bu(\bx,t)$, thickness $\eta(\bx, t)$, buoyancy $b(\bx,t)$, 
and free surface elevation $\alpha\zeta = \eta - h$, with mean depth $h(\bx)$, given in \eqref{trsw-eqns} by
\begin{equation}   
\frac{D}{Dt} \bu+ \frac{1}{{\rm Ro}}f\zh\times\bu   
= - \frac{\alpha}{{\rm Fr}^2}\nabla\big((1+\mathfrak{s}b) \zeta\big) + \frac{\mathfrak{s}}{2\,{\rm Fr}^2}(\zeta-h)\nabla b \, ,   
\qquad 
\frac{\partial \eta}{\partial t} + \nabla\cd (\eta\bu) = 0\, ,   
\qquad
\frac{D}{Dt}b = 0\,,
\label{trsw-eqns-qg}  
\end{equation}
with boundary conditions in \eqref{trsw-bdy}. In order to find an asymptotic balance among these equations, a number of assumptions are necessary. First, we assume that the free surface elevation $\alpha\zeta = \eta - h$ is small,  that is $\alpha\ll1$. Second, in line with the Boussinesq approximation in three dimensional fluids, we assume that the buoyancy stratification is also small, meaning that the stratification parameter $\mathfrak{s}\ll 1$ . Third, we assume that the all dimensionless numbers have equal magnitude
\begin{equation}
\mathcal{O}(\alpha) = \mathcal{O}(\mathfrak{s}) = \mathcal{O}({\rm Ro}) = \mathcal{O}({\rm Fr}).
\label{trsw-order}
\end{equation}
Now, in the QG approximation one also assumes the gradients of bathymetry and Coriolis parameter are small, of order $\cal O({\rm Ro})$. Taken together, this amounts to the following set of assumptions:
\begin{equation}
\begin{aligned}
f(\bx) &= 1 + {\rm Ro}\,f_1(\bx),\\
h(\bx) &= 1 + {\rm Ro}\,h_1(\bx).
\end{aligned}
\label{qg-assumptions}
\end{equation} 
The beta plane approximation is included in the notation used in \eqref{qg-assumptions}, when $\beta f_0^{-1} = \mathcal{O}({\rm Ro})$ and $f_1(\bx) = y$. These assumptions were also made in derivation of this model in \cite{holm2019stochastic} and they are sufficient to derive the TRSW balance relation, which we will show now. Note that since all dimensionless numbers have the same magnitude as in \eqref{trsw-order}, we can continue the analysis with a single small parameter $\epsilon \ll 1$. First, we multiply the momentum equation in \eqref{trsw-eqns-qg} by $\epsilon$, then we substitute the assumptions \eqref{qg-assumptions} into the momentum equation to find
\begin{equation}
\begin{aligned}
%{\rm Ro}\,\frac{D}{Dt} \bu+ \zh\times\bu  + {\rm Ro}\,f_1\zh\times\bu   
%&= - \frac{\epsilon\,{\rm Ro}}{{\rm Fr}^2}\nabla((1+\epsilon b_1)\zeta_1) + \frac{\epsilon\,{\rm Ro}}{2{\rm Fr}^2}(\epsilon\zeta_1-1 - {\rm Ro}\,h_1)\nabla b_1\\
%\zh\times\bu  + {\rm Ro}\,\frac{D}{Dt} \bu+ {\rm Ro}\,f_1\zh\times\bu   
%&= - \nabla\zeta_1 - \epsilon \nabla(b_1\zeta_1) + \frac{1}{2}(\epsilon\zeta_1-1 - {\rm Ro}\,h_1)\nabla b_1\\
%\zh\times\bu  + {\rm Ro}\,\frac{D}{Dt} \bu+ {\rm Ro}\,f_1\zh\times\bu   
%&= - \nabla\zeta_1 - \frac{1}{2}\nabla b_1 - \epsilon b_1\nabla\zeta_1 - \epsilon\zeta_1\nabla b_1 + \frac{1}{2}(\epsilon\zeta_1 - {\rm Ro}\,h_1)\nabla b_1\\
\zh\times\bu  + \epsilon\,\frac{D}{Dt} \bu+ \epsilon\,f_1\zh\times\bu   
&= - \nabla\zeta_1 - \frac{1}{2}\nabla b - \epsilon b\nabla\zeta - \frac{\epsilon}{2}(\zeta + h_1)\nabla b_1
\,.\label{ordering}
\end{aligned}
\end{equation}
Thus, equation \eqref{ordering} implies the following relation
\begin{equation}   
\zh\times\bu = - \nabla\zeta - \frac{1}{2}\nabla b + \mathcal{O}(\epsilon)
\,.
\label{trsw-bal}  
\end{equation}
By operating with $\zh\times$ on the TRSW balance relation \eqref{trsw-bal}, we find the defining expression for the divergence-free thermal geostrophic velocity field, denoted as $\bu_{TG}$, 
\begin{equation}
\bu_{TG} = \zh\times\nabla\left(\zeta + \frac{1}{2}b\right) =: \zh\times\nabla\psi\,.
\label{trsw-tg}
\end{equation}
Here, $\psi(\bx,t)$ plays the role of the \emph{stream function} for the divergence-free leading order thermal geostrophic velocity vector field $\bu_{TG}$. Relation \eqref{trsw-bal} allows the total fluid velocity to be represented as the sum of the leading order thermal geostrophic velocity $\bu_{TG}$ and higher order terms,
\begin{equation}
\bu = \bu_{TG} + \epsilon\boldsymbol \upsilon %= \bu_{TG} + \mathcal O(\epsilon).
\label{trsw-geostrophicdecomposition}
\end{equation}
In particular, \eqref{trsw-geostrophicdecomposition} implies that the difference, referred to as the ageostrophic component of the velocity field, satisfies $\bu-\bu_{TG} = \epsilon \boldsymbol \upsilon = \mathcal O(\epsilon)$. This decomposition of the velocity field into a leading order  divergence-free part plus higher order parts is similar to the Helmholtz decomposition, except the  divergence-free component is allowed to have both a leading order part and a higher order part, while the rotation free component has only higher order parts.

\subsection{Moving into the balanced frame}
One may transform the thermal rotating shallow water equations in \eqref{trsw-eqns} into a time-dependent local frame moving with the thermal geostrophically balanced velocity, $\bu_{TG}(\bx,t)$, by inserting the decomposition \eqref{trsw-geostrophicdecomposition} into the Lagrangian for the TRSW equations \eqref{lag:TRSW},
%\begin{equation}
%\begin{aligned}
%\ell_{TRSW} &= \int \frac{1}{2}\eta |\bu_{TG} + \epsilon\boldsymbol \upsilon|^2 + \frac{1}{\epsilon}\eta(\bu_{TG}+\epsilon\boldsymbol\upsilon)\cdot\bR - \frac{1}{2\epsilon^2} b(\eta^2-2\eta h)\,dx\,dy \\
%&= \int \frac{1}{2}\eta |\bu_{TG}|^2 + \epsilon\eta\bu_{TG}\cdot\boldsymbol\upsilon + \frac{1}{2}\epsilon^2\eta|\boldsymbol \upsilon|^2 + \frac{1}{\epsilon}\eta(\bu_{TG}+\epsilon\boldsymbol\upsilon)\cdot\bR - \frac{1}{2\epsilon^2} b(\eta^2-2\eta h)\,dx\,dy \\
%&= \int \frac{1}{2}\epsilon^2\eta |\boldsymbol \upsilon|^2 + \epsilon \boldsymbol \upsilon\cdot\left(\eta\bu_{TG}+\frac{1}{\epsilon}\eta\bR\right) + \bu_{TG}\cdot\left(\eta\bu_{TG}+\frac{1}{\epsilon}\eta\bR\right)-\frac{1}{2}\eta|\bu_{TG}|^2 - \frac{1}{2\epsilon^2} b(\eta^2-2\eta h)\,dx\,dy
%\end{aligned}
%\label{lag:trswbalanced}
%\end{equation}

\begin{equation}
\begin{aligned}
\ell_{TRSW}(\mathbf{u},b,\eta)
&= \int_{\cal D} \frac{1}{2}\eta |\bu|^2 + \frac{1}{\epsilon}\eta\bu \cdot\bR - \frac{1}{2\epsilon^2} (1+\epsilon b)(\eta^2-2\eta h)\,dx\,dy \\
&= \int_{\cal D} \epsilon\eta \left( \frac{\epsilon}{2} |\sym{\upsilon}|^2 + \sym{\upsilon}\cdot\bu_{TG} + \frac{1}{\epsilon}\sym{\upsilon}\cdot\bR\right) 
+ \eta \left(  \frac{1}{2}|\bu_{TG}|^2 + \frac{1}{\epsilon} \bu_{TG}\cdot\bR\right) 
\\&\hspace{2cm}- \frac{1}{2\epsilon^2} (1+\epsilon b)(\eta^2-2\eta h)\,dx\,dy\,.
%\hbox{Upon inserting \eqref{trsw-tg} and  \eqref{trsw-geostrophicdecomposition}}
%&= \int A(\mathbf{u},b,\eta)\delta \eta + \bfm(\mathbf{u},b,\eta)\cdot \delta \bu + C(\mathbf{u},b,\eta)\delta b \,dx\,dy\,dt \\
%&= \int \big( A\delta \eta + \bfm\cdot \delta\sym{\upsilon} +  C\delta b \big)
%\,dx\,dy\,dt
\end{aligned}
\label{lag:trswbalanced}
\end{equation}
%where 
%\begin{equation}
%\begin{aligned}
%\frac{\delta \ell}{\delta \sym{\upsilon}} = \bfm & := \eta\Big(\epsilon \sym{\upsilon} + \bu_{TG}(b,\eta) + \frac{1}{\epsilon}\bR(\bx)\Big) 
%= \eta\Big(\bu + \frac{1}{\epsilon}\bR(\bx)\Big) \,, \\
%\frac{\delta\ell}{\delta \eta} = A 
%& :=   \frac{1}{2} |\bu|^2 + \frac{1}{\epsilon}\bu\cdot \bR - \frac{1}{\epsilon^2} b (\eta - h) - \zh\cdot{\rm curl}\bfm\,, \\
%\frac{\delta\ell}{\delta b} = C &:= - \frac{1}{2\epsilon^2} (\eta^2 - 2h\eta) - \zh\cdot{\rm curl}\bfm\,,
%\end{aligned}
%\label{lag:trswbalanced1}
%\end{equation}
%and we have 
One can apply Hamilton's variational principle to the Lagrangian \eqref{lag:trswbalanced} $0 = \delta S$ with $S = \int_{t_1}^{t_2} \ell_{TRSW} dt$. We substitute the velocity decomposition in \eqref{trsw-geostrophicdecomposition} to define the variation $\delta \bu = \delta \bu_{TG} + \epsilon \delta \sym{\upsilon}$ with
\begin{equation}
\delta \bu_{TG} =  \zh\times\nabla\left(\delta \eta + \frac{1}{2}\delta b\right) .
\label{trsw-tg1}
\end{equation}
Hamilton's principle then  yields the Euler--Poincar\'e equation \eqref{eq:EPeq-circ} in the form,
\begin{equation}
\frac{\partial}{\partial t}\left(\frac{1}{\eta}\frac{\delta \ell}{\delta \sym{\upsilon}}\right)
+ (\bu\cdot\nabla)\left(\frac{1}{\eta}\frac{\delta \ell}{\delta \sym{\upsilon}}\right)
+ (\nabla \bu)\cdot\left(\frac{1}{\eta}\frac{\delta \ell}{\delta \sym{\upsilon}}\right)
= -\,\frac{1}{\eta}\frac{\delta \ell}{\delta b}\nabla b + \nabla\frac{\delta\ell}{\delta \eta}
\,.\label{eq:EPeq-circ1}
\end{equation}
Thus, the relative motion equation for TRSW dynamics in the frame moving with the thermal geostrophic balance velocity $\bu_{TG}(\bx,t)$ in \eqref{trsw-tg} keeps its Euler--Poincar\'e form \eqref{eq:EPeq-circ}. Upon eliminating $\partial_t \bu_{TG}$ in \eqref{eq:EPeq-circ1} by using the advection equations for $(b,\eta)$ in \eqref{trsw-eqns-qg}, the system closes and thereby transforms the TRSW equations into the new variables $(\boldsymbol \upsilon,b,\eta)$ in the reference frame moving with velocity $\bu_{TG}(\bx,t)$. 

\paragraph{\emph{Stationary} thermal geostrophic balance as ``mean dynamic topography''.} To a good approximation, much of upper ocean dynamics is well-approximated by a \emph{mean dynamic topography} (MDT), which is monitored continuously with in situ instruments and satellites, see, e.g., \cite{maximenko2009mean}. Ocean dynamics is then envisioned as time-dependent variations in the steady moving frame of the MDT.  To apply this idea to TRSW dynamics, we envision TRSW dynamics as taking place in the moving reference frame defined by a time-independent mean thermal geostrophic velocity $\overline{\bu}_{TG}(\bx)$. In this situation, the Kelvin theorem \ref{Kel-trsw} for deterministic TRSW derived from equation  \eqref{eq:EPeq-circ1} takes the following form. \medskip

\begin{theorem}[Kelvin theorem for deterministic TRSW in a stationary balanced frame]\label{Kel-thmbalanceTRSW}
The deterministic TRSW equations \eqref{trsw-eqns} imply the following Kelvin circulation law in a stationary TG balanced frame moving with time-independent velocity $\overline{\bu}_{TG}(\bx)$, 
\begin{equation}
\frac{d}{dt}\oint_{c(\sym{\upsilon})} \hspace{-2mm}\Big(\epsilon\,\sym{\upsilon} +\bu_{TG}(\bx) 
+ \frac{1}{\epsilon}\bR(\bx)\Big) \cdot d\bx 
= \frac{1}{2\epsilon}  \oint_{c(\sym{\upsilon})} \hspace{-2mm} (\epsilon\zeta-h) \nabla b \cdot d\bx
= \frac{1}{2\epsilon}  \int\!\!\int_{\partial S = c(\sym{\upsilon})}
 \hspace{-2mm}
\zh\cdot\nabla(\epsilon\zeta-h) \times\nabla b\, dx\,dy
\,,
\label{Kel-trsw-frame}   
\end{equation}
where $c(\sym{\upsilon})$ is any closed loop moving with horizontal fluid velocity $\sym{\upsilon}(\bx,t)$ relative to the frame of motion whose velocity is $\overline{\bu}_{TG}(\bx)+ \epsilon^{-1}\bR(\bx)$ in two horizontal dimensions. 
\end{theorem}\medskip

\begin{remark}
Thus, in the mean thermal geostrophic balance frame, the frame velocity $\overline{\bu}_{TG}(\bx)$ simply adds another contribution to the momentum per unit mass. In turn, this contributes an additional `Coriolis' force in the dynamics of the relative velocity $ \epsilon\boldsymbol \upsilon= \bu - \bu_{TG} $. The corresponding SALT version in this case would simply replace $\mathbf{u}(\mathbf{x},t)$ by $\sym{\upsilon}(\mathbf{x},t)$ as the drift velocity of the stochastic vector field ${\sf d}\boldsymbol \chi_{t}$ defined in \eqref{def:chi}.
\end{remark}
In the next sections, we will consider the thermal versions of two of the classic GFD approximations of RSW developed previously in the absence of buoyancy. Namely, we will consider thermal versions of the Eliassen approximation and the quasigeostrophic approximation. 

\subsection{The Eliassen approximation}
The starting point in deriving the thermal Eliassen approximation is the Lagrangian for the thermal rotating shallow water equations \eqref{lag:TRSW}
\begin{equation}
\ell_{TRSW} = \int \frac{1}{2}\eta |\bu|^2 + \frac{1}{{\rm Ro}}\eta\bu\cdot\bR - \frac{1}{2\,{\rm Fr}^2} (1+\mathfrak{s}b)(\eta^2-2\eta h)\,dx\,dy \,.
\label{lag:trswl1}
\end{equation}
Since the thermal geostrophic velocity field \eqref{trsw-tg} is  divergence-free, it is useful to transform the velocity variables inside the Lagrangian to vorticity and divergence by using the Helmholtz decomposition. The forward transformation is $(\bu,\eta,b)\mapsto(\omega,D,\eta,b)$, which amounts to 
\begin{equation}
\begin{aligned}
\omega &= \zh\cdot\nabla\times\bu,\\
D &= \nabla\cdot\bu,\\
\eta &= \eta,\\
b &= b.
\end{aligned}
\end{equation}
The inverse transformation is unique if the kernel of the Laplacian is trivial, which is the same as saying that there are no harmonic functions for the domain $\cal D$ for the boundary conditions on $\partial \cal D$. In this case, the inverse transformation $(\omega,D,\eta,b)\mapsto(\bu,\eta,b)$ is given by
\begin{equation}
\begin{aligned}
\bu &= \zh\times\nabla\Delta^{-1}\omega + \nabla\Delta^{-1} D,\\
\eta &= \eta,\\
b &= b.
\end{aligned}
\label{Invert-Hodge}
\end{equation}
The inverse transformation \eqref{Invert-Hodge} uniquely defines the vector $\bu$ in terms of its divergence and curl. The inverse of the Laplacian can be interpreted in terms of the appropriate Green's function in two dimensions. Boundary conditions need to be dealt with carefully when taking the Green's function approach. Assuming that this is the case, the symbol $\Delta^{-1}$ denotes the correct Green's function. For doubly periodic boundary conditions, the Green's function for the Laplacian takes the form
\begin{equation}
\Delta^{-1} F = -\frac{1}{2\pi}\int \ln\|\bx - \bx'\| F(\bx') d\bx'.
\end{equation}
Changing variables using the classical Helmholtz decomposition leads to the following formulation of the Lagrangian for thermal rotating shallow water
\begin{equation}
\begin{aligned}
\ell_{TRSW} &= \int \frac{1}{2} \eta\left|\nabla\Delta^{-1}\omega\right|^2 +\eta J\left(\Delta^{-1}\omega,\Delta^{-1}D\right) + \frac{1}{2}\eta|\nabla\Delta^{-1}D|^2
\\&\quad 
+ \frac{1}{{\rm Ro}}\eta\left(\zh\times\nabla\Delta^{-1}\omega + \nabla\Delta^{-1} D\right)\cdot\bR - \frac{1}{2\,{\rm Fr}^2} (1+\mathfrak{s}b)(\eta^2-2\eta h)\,dx\,dy \,.
\end{aligned}
\end{equation}
It should be noted that the vorticity and divergence are not orthogonal in a weighted $L^2$ space, so the Jacobian term $\int\eta J\left(\Delta^{-1}\omega,\Delta^{-1}D\right)\,dx\,dy\neq 0$. The Jacobian term does vanish in the standard $L^2$ space, in which $\int J\left(\Delta^{-1}\omega,\Delta^{-1}D\right)\,dx\,dy = 0$. By means of the ordering $\mathcal{O}(\alpha) = \mathcal O(\mathfrak{s}) = \mathcal O({\rm Ro}) = \mathcal O({\rm Fr})$ made in \eqref{ordering} and the decomposition of $\bu$ into a thermal geostrophic part and a higher order part \eqref{trsw-geostrophicdecomposition}, one can take the following asymptotic expansions for the vorticity and the divergence
\begin{equation}
\begin{aligned}
\omega &= \omega_0 + \epsilon \omega_1 + o(\epsilon),\\
D &= \epsilon D_1 + o(\epsilon).
\end{aligned}
\label{exp:lagtqg}
\end{equation}
The thermal geostrophic balance implies a decomposition of the velocity field in terms a leading order  divergence-free component and a higher order general component. This means that one can identify $\omega_0$ with the curl of the thermal geostrophic velocity field \eqref{trsw-tg}, but it also means that one must keep a higher order vorticity term around. Substituting \eqref{exp:lagtqg} into the Lagrangian yields
\begin{equation}
\begin{aligned}
\ell_{TRSW} &= \int \frac{1}{2} \eta\left|\nabla\Delta^{-1}\left(\omega_0 + \epsilon\omega_1\right)\right|^2 
\\&\quad 
+\eta J\left(\Delta^{-1}\left(\omega_0 + \epsilon\omega_1\right),\Delta^{-1}\epsilon D_1\right) + \frac{1}{2}\eta|\nabla\Delta^{-1}\epsilon D_1|^2
\\&\quad 
+ \frac{1}{\epsilon}\eta\left(\zh\times\nabla\Delta^{-1}\left(\omega_0 + \epsilon\omega_1\right) + \nabla\Delta^{-1}\epsilon D_1\right)\cdot\bR
\\&\quad 
- \frac{1}{2\epsilon^2}(1+\epsilon b)(\eta^2-2\eta h)\,dx\,dy\, + o(\epsilon).
\end{aligned}
\end{equation}
By expanding and collecting all terms that are of higher order than $\mathcal O(1)$, the Lagrangian can be written as
\begin{equation}
\begin{aligned}
\ell_{TRSW} &= \int \frac{1}{2} \eta\left|\nabla\Delta^{-1}\omega_0\right|^2 
+ \frac{1}{\epsilon}\eta\left(\zh\times\nabla\Delta^{-1}(\omega_0+\epsilon\omega_1) + \epsilon\nabla\Delta^{-1} D_1\right)\cdot\bR 
- \frac{1}{2\epsilon^2}b(\eta^2-2\eta h)\,dx\,dy 
+ o(1).
\end{aligned}
\end{equation}
We now use the fact that the velocity field also decomposes into a thermal geostrophic part and an ageostrophic part and apply the inverse transformation to recover the original fluid variables.  This yields 
\begin{equation}
\ell_{TL1} = \int \frac{1}{2}\eta|\bu_{TG}|^2 + \frac{1}{\epsilon}\eta\bu\cdot\bR - \frac{1}{2\epsilon^2}(1+\epsilon b)(\eta^2-2\eta h)\,dx\,dy \,.
\end{equation}
The subscript $TL1$ refers to the thermal $L1$ model, which is an extension of \cite{salmon1983practical} to include horizontal variations in buoyancy and bathymetry. However, at this stage we do not have enough information to execute the variational principle. To obtain the required information, we will introduce a higher order term which completes the Lagrangian, by enabling us to vary with respect to the full velocity field $\bu$, interpreted as a Lagrange multiplier
\begin{equation}
\ell_{TL1} = \int \bu\cdot\left(\eta\bu_{TG}+\frac{1}{\epsilon}\eta\bR \right) - \frac{1}{2}\eta|\bu_{TG}|^2 - \frac{1}{2\epsilon^2}(1+\epsilon b)(\eta^2-2\eta h)\,dx\,dy + \mathcal O(\epsilon).
\label{lag:tl1}
\end{equation}
Equivalently, one can truncate the Lagrangian for TRSW rewritten in the balanced frame \eqref{lag:trswbalanced} at $\mathcal{O}(1)$ to obtain this Lagrangian. This emphasises the fact that a component of the velocity field which can be expressed in terms of the other variables in the problem can be used to change the reference frame. At this point, several important questions arise. How does one take variations of this Lagrangian? Can the equation for the Lagrange multiplier $\bu$ be found? To answer these questions, we use the methods of \cite{allen1996extended}.
\subsection{The Allen-Holm approach}
The \cite{allen1996extended} approach is based on the following observation. A Lagrangian $\ell$ leads to the corresponding Hamiltonian $\hslash$ via the Legendre transform (which is assumed to be invertible for the given $\ell$)
\begin{equation}
\hslash(\mathbf{m},\eta) = \int\left(\mathbf{m}-\frac{\delta \ell}{\delta \bu}\right)\cdot \bu + \overline{\mathcal E}(\bu, \eta, b, \nabla\eta, \nabla b, {\rm etc}.),
\label{lt:ah}
\end{equation} 
where $\overline{\mathcal E}$ can be interpreted as the energy density. The momentum density $\mathbf{m}$ is given in terms of the other fluid variables by the condition $\delta \hslash/\delta \bu = 0$, where $\hslash$ is the Hamiltonian defined by the Legendre transform in \eqref{lt:ah}. In the Legendre transform, the fluid velocity $\bu$ appears as a Lagrange multiplier which enforces the relation of $\mathbf{m}$ to the other fluid variables as a dynamically preserved constraint. This definition is usually taken for granted, but in what follows, we shall model the momentum density as a \emph{prescribed} function of the other fluid variables. This means that we will define
\begin{equation}
\mathbf{m} = \overline{\mathbf{m}}(\eta,b,\nabla\eta,\nabla b,{\rm etc.})=:\overline{\mathbf{m}}[\eta,b].
\label{ah-momentum}
\end{equation}
In this type of modelling, it is necessary to have the explicit enforcement of the momentum definition \eqref{ah-momentum}, both as a constraint as well as a means of determining the fluid velocity for the model by using Lagrange multipliers. We rearrange the Lagrangian in \eqref{lag:trswl1} using $\overline{\mathbf{m}}$ as the momentum density, defined by
\begin{equation}
\overline{\mathbf{m}} := \frac{\delta \ell_{TRSW}}{\delta \bu} = \eta\bu + \frac{1}{{\rm Ro}}\eta\bR.
\label{trsw-mbar-def}
\end{equation}
The Lagrangian in \eqref{lag:trswl1} can then be written as
\begin{equation}
\ell_{TRSW} = \int_{\cal D} \overline{\mathbf{m}}\cdot\bu - \frac{1}{2}\eta|\bu|^2 - \frac{1}{2\,{\rm Fr}^2}(1+\mathfrak{s} b)(\eta^2 - 2\eta h) \,dx\,dy\,.
\label{trsw-lagAH}
\end{equation} 
%According to \eqref{trsw-mbar-def} and \eqref{trsw-lagAH}, $\overline{\mathbf{m}}/\eta$ acts to shift the velocity of the inertial Eulerian reference frame by an amount $\bR/{\rm Ro}$. Changes of reference frame are common in geophysical fluid dynamics, see for instance \cite{holm2020variational}, where a wave--mean flow decomposition leads to a change of frame. 
Substitution of the decomposition of the velocity field into its geostrophic and ageostrophic components in \eqref{trsw-geostrophicdecomposition} with $\bu_{TG}$ defined in \eqref{trsw-tg} into the Lagrangian \eqref{trsw-lagAH} now leads, without approximation, to the following Lagrangian, which is \emph{linear} in the velocity $\bu$,
\begin{equation}
\ell_{TRSW}
=
 \int_{\cal D} \bu\cdot\left(\eta\bu_{TG} + {\rm Ro}\,\eta\bu_A + \frac{1}{{\rm Ro}}\eta\bR\right) - \frac{1}{2}\eta|\bu_{TG} + {\rm Ro}\,\bu_A|^2 - \frac{1}{2\,{\rm Fr}^2} (1+\mathfrak{s}b)(\eta^2-2\eta h) \,dx\,dy.
\label{trsw-lagAH2}
\end{equation}
In line with the ordering scheme $\mathcal O(\alpha) = \mathcal{O}(\mathfrak{s}) = \mathcal O({\rm Ro}) = \mathcal O({\rm Fr})$, we formulate the Lagrangian in terms of a single parameter $\epsilon$. When $\epsilon \ll 1$, one could simply drop the $\bu_A$ terms in the Lagrangian to obtain
\begin{equation}
\ell_{TL1}
=
 \int_{\cal D} \bu\cdot\left(\eta\bu_{TG} + \frac{1}{\epsilon}\eta\bR \right) - \frac{1}{2}\eta|\bu_{TG}|^2 - \frac{1}{2\epsilon^2} (1+\epsilon b)(\eta^2-2\eta h) \,dx\,dy.
\label{trsw-lagAH4}
\end{equation}
One recalls that this is the Lagrangian obtained in \eqref{lag:tl1} by using the Helmholtz decomposition to decompose $\bu$ into vorticity and divergence. 
\bigskip

\begin{remark} \label{rem:strictasymptotics}
Since we will use \eqref{trsw-tg} as our definition for $\bu_{TG}$, we should keep in mind that according to strict asymptotics the potential energy term should also be expanded using the same assumptions that led to \eqref{trsw-tg}. By keeping the Lagrangian in the form \eqref{trsw-lagAH4} we have included higher order terms, but not all of them, since we have truncated the kinetic energy. In terms of strict asymptotics this means that we do not have a balance among terms in the Lagrangian. A benefit of not expanding the potential energy at this stage is that the variational derivatives can be taken in the usual way and are thus closer to the variational derivatives of the TRSW system. 
\end{remark}
\bigskip

\noindent The approximate Lagrangian \eqref{trsw-lagAH4} is also linear in the velocity $\bu$, since it has the form
\begin{equation}
\ell_{TL1} = \int \overline{\mathbf{m}}[\eta,b]\cdot\bu-\overline{\mathcal E}[\eta,b]\,dx\,dy,
\label{AH-ansatz}
\end{equation}
with
\begin{equation}
\overline{\mathbf{m}}[\eta,b]=\frac{\delta \ell_{TL1}}{\delta\bu}= \frac{1}{\epsilon}\eta\bR + \eta\bu_{TG} \qquad \text{ and } \qquad \overline{\mathcal E}[\eta,b] = \frac{1}{2}\eta|\bu_{TG}|^2 + \frac{1}{2\epsilon^2} (1+\epsilon b)(\eta^2-2\eta h) .
\end{equation}
Note that $\bu_{TG}$ can be expressed in terms of $\eta$ and $b$. At this stage, in \cite{allen1996extended}, the next step after having obtained the Lagrangian in the form above would have been to take the Legendre transformation and obtain the Hamiltonian. Then, by requiring the first variation of the Hamiltonian to vanish, one would obtain the equations of motion. In \cite{holm1998euler} it was shown, however, that one can obtain the same equations of motion by applying the variational principle on the Lagrangian side, by means of the Euler--Poincar\'e theorem. The first variation of the Lagrangian is given by
\begin{equation}
\begin{aligned}
\delta \ell_{TL1} &= \int_{\cal D} \left(\eta\bu_{TG} + \frac{1}{\epsilon}\eta\bR\right)\cdot \delta \bu + \Big(\eta(\bu-\bu_{TG}) \Big)\cdot\delta\bu_{TG}
\\&\qquad
+\left(\bu_{TG}\cdot\bu+\frac{1}{\epsilon}\bu\cdot\bR - \frac{1}{2}|\bu_{TG}|^2 - \frac{1}{\epsilon}(1+\epsilon b)\zeta\right) \delta \eta
\\&\qquad
+\left(-\frac{1}{2\epsilon}(\eta^2-2\eta h)\right)\delta b \,dx\,dy\,.
\end{aligned}
\label{1st-var-L1}
\end{equation}
From the definition of $\bu_{TG}$ in \eqref{trsw-tg}, we now substitute
\begin{equation}
\delta\bu_{TG} = \frac{1}{\epsilon}\zh\times\nabla\delta\eta + \frac{1}{2}\zh\times\nabla \delta b
\,.\label{uTG-var}
\end{equation}
Integration by parts in \eqref{1st-var-L1} then yields
\begin{equation}
\begin{aligned}
\delta \ell_{TL1} &= \int_{\cal D} \left(\eta\bu_{TG} + \frac{1}{\epsilon}\eta\bR\right)\cdot \delta \bu
\\&\quad
+\left(\bu_{TG}\cdot\bu + \frac{1}{\epsilon}\bR\cdot\bu - \frac{1}{2}|\bu_{TG}|^2 - \frac{1}{\epsilon}(1+\epsilon b)\zeta - \frac{1}{\epsilon}\zh\cdot\nabla\times\Big(\eta(\bu-\bu_{TG})\Big)\right) \delta \eta
\\&\quad
+\left(-\frac{1}{2\epsilon}(\eta^2-2\eta h) - \frac{1}{2}\zh\cdot\nabla\times\Big(\eta(\bu-\bu_{TG})\Big)\right)\delta b \,dx\,dy
\\&\quad 
- \oint_{\partial\cal D}\zh\times\Big(\eta(\bu-\bu_{TG})\Big)\left(\frac{1}{\epsilon}\delta\eta + \frac{1}{2}\delta b\right)\cdot d\bx
\,.\end{aligned}
\label{1st-var-L1-byParts}
\end{equation}
The integral over the boundary vanishes provided that the ageostrophic velocity field has no tangential component on the boundary. This boundary condition is satisfied since the ageostrophic velocity field can be represented by a velocity potential, as in \eqref{Invert-Hodge}. By means of the Euler--Poincar\'e theorem, we find the equations of motion given in a form first proposed by Eliassen \cite{eliassen1949quasi} as
\begin{equation}
\begin{aligned}
\frac{\partial}{\partial t}\bu_{TG} - \bu\times\left(\nabla\times\Big(\bu_{TG}+\frac{1}{\epsilon}\bR\Big)\right) 
+ \nabla\left(\frac{1}{\epsilon}(1+\epsilon b)\zeta + \frac{1}{2}|\bu_{TG}|^2 
+ \frac{1}{\epsilon} B\right) - \left(\frac{1}{2\epsilon}(\epsilon\zeta-h)
+\frac{1}{2\eta}B\right)\nabla b &= 0,
\\
\frac{\partial}{\partial t}\eta + \nabla\cdot(\eta\bu) = 0, 
\hspace{2cm}
\frac{\partial}{\partial t} b + \bu\cdot\nabla b &= 0.
\end{aligned}
\label{thermal-l1}
\end{equation}
The notation $B$ in the first of these equations is defined by
\begin{equation}
B := \zh\cdot\nabla\times\Big(\eta(\bu-\bu_{TG})\Big).
\label{B-def}
\end{equation}
The function $B$ keeps track of effects that are generated by higher order vorticity terms, since $B$ includes the curl of the velocity difference $\bu-\bu_{TG}$ , which is of order $\mathcal O(\epsilon)$. These higher order vorticity terms will contribute in the Kelvin circulation theorem \ref{Kel-thm-TL1} arising from equations \eqref{thermal-l1}. Equations \eqref{thermal-l1} extend Salmon's $L1$ model \cite{salmon1983practical} to include horizontal buoyancy variations and bottom topography. The boundary conditions carry over from thermal rotating shallow water and are given by
\begin{equation} 
\mathbf{\hat{n}}\cd \bu=0
\quad\hbox{and}\quad
\mathbf{\hat{n}}\times\nabla b = 0
\quad\hbox{on the boundary $\partial{\cal D}$.}
\label{eq:bdytl1}  
\end{equation}
Having been derived from the Euler--Poincar\'e variational principle \cite{holm1998euler}, the deterministic TL1 equations in \eqref{thermal-l1} satisfy the following Kelvin--Noether circulation theorem.
\medskip

\begin{theorem}[Kelvin theorem for the deterministic TL1 model]\label{Kel-thm-TL1}
The deterministic TL1 equations \eqref{thermal-l1} imply the following Kelvin circulation law
\begin{equation}
\begin{aligned}
\frac{d}{dt}\oint_{c(\bu)} \left(\bu_{TG} + \frac{1}{\epsilon}\bR\right)\cdot d\bx &= \oint_{c(\bu)} \left( \frac{1}{2\epsilon}(\epsilon\zeta-h)+\frac{1}{2\eta}B\right)\nabla b\cdot d\bx\\
&= \int\!\!\int_{\partial S = c(\bu)} \zh\cdot\nabla\left( \frac{1}{2\epsilon}(\epsilon\zeta-h)+\frac{1}{2\eta}B\right)\times\nabla b \, dx\,dy\,.
\end{aligned}
\label{l1-kelvin}
\end{equation} 
\end{theorem}

\begin{proof}
This result follows from the Kelvin--Noether theorem  \ref{thm:Kelvin} for Euler--Poincar\'e fluid equations. 
\end{proof}

\begin{remark}\label{rem:kel-tl1}
The Kelvin circulation theorem \ref{Kel-thm-TL1} for the deterministic TL1 model  \eqref{thermal-l1} implies that the misalignment of the horizontal gradients of the free surface elevation and the bathymetry with the horizontal gradient of the buoyancy generates circulation. This result is similar to the corresponding Kelvin circulation theorem \ref{Kel-thmTRSW} for the deterministic thermal rotating shallow water (TRSW) model. An additional contribution to the generation of circulation relative to theorem \ref{Kel-thmTRSW} is made by the misalignment of the gradient of the quantity $(B/\eta)$ defined in \eqref{B-def} with the gradient of the buoyancy. This additional contribution is due to misalignment of the horizontal gradients of the ageostrophic vorticity and the buoyancy.
\end{remark}
\medskip

\noindent The potential vorticity $q$ for TL1 is defined as
\begin{equation}
q = \frac{1}{\eta}\left(\omega + \frac{1}{\epsilon}f\right) = \frac{1}{\eta}\left(\Delta\Big(\zeta + \frac{1}{2}b\Big)+\frac{1}{\epsilon} + f_1\right)
. 
\label{def:pvtL1}
\end{equation} 
Note that the potential vorticity $q$ in \eqref{def:pvtL1} contains a term in the Coriolis parameter which is order $\mathcal O(\epsilon^{-1})$. This feature will become important in the asymptotic expansion of $q$ later, in deriving the thermal QG model at the beginning of section \ref{sec:TQGmodel}. In the presence of buoyancy, the evolution of potential vorticity along Lagrangian fluid trajectories is not conserved. Instead, potential vorticity is generated, as indicated in the circulation theorem \eqref{l1-kelvin} via misalignment of gradients in 
\begin{equation}
\frac{\partial}{\partial t}q + \bu\cdot\nabla q = \frac{1}{\eta}\zh\cdot\nabla\left( \frac{1}{2\epsilon}(\epsilon\zeta-h)+\frac{1}{2\eta}B\right)\times\nabla b.
\label{eq:l1pv}
\end{equation}
Although the potential vorticity is not conserved along Lagrangian fluid trajectories, the TL1 equations do preserve energy, as well as an infinity of integral conservation laws involving buoyancy and potential vorticity.
\paragraph{Conservation laws for deterministic TL1.} The deterministic TL1 equations \eqref{thermal-l1} conserve the energy
\begin{equation}
\mathcal E_{TL1}(\bu_{TG},\eta,b) =  \frac12 \int_{\cal D} \eta |\bu_{TG}|^2 + \frac{1}{\epsilon^2}(1+\epsilon b)(\eta^2-2\eta h)\,dx\,dy\,.
\label{eq:erg-tL1}
\end{equation}
Equations \eqref{thermal-l1} also conserve an infinity of integral conservation laws, determined by two arbitrary differentiable functions of buoyancy $\Phi(b)$ and $\Psi(b)$ as
\begin{equation}
C_{\Phi,\Psi} = \int_{\cal D} \eta\Phi(b) + \varpi\Psi(b) \,dx\,dy = \int_{\cal D} \Big( \Phi(b) + q \Psi(b)\Big)\eta\,dx\,dy\,,
\label{eq:casimirstL1}
\end{equation}
where $\omega$ and $q$ are defined in equation \eqref{def:pvtL1}. Notice that this family of integral conservation laws for the TL1 equations has the same form as the family of integral conserved quantities for the TRSW equations, defined in \eqref{eq:casimirstrsw}. The proof that $C_{\Phi,\Psi}$ is conserved in time, for any choice of differentiable $\Phi$ and $\Psi$, follows from a direct calculation involving the boundary conditions \eqref{eq:bdytl1}.
The integral conserved quantities in equations \eqref{eq:erg-tqg} and \eqref{eq:casimirstL1} are associated with smooth transformations which leave invariant the Eulerian fluid quantities in the Lagrangian. 
As with the TRSW equations, upon introducing stochasticity via the Euler--Poincar\'e theorem \ref{thm:SEP}, the latter conservation laws persist. However, energy conservation does not persist because the stochastic Lagrangian depends explicitly on time through the Brownian motion. 

In order to use the TL1 equations \eqref{thermal-l1} as a predictive model, one needs to be able to determine the Lagrange multiplier $\bu$ from the other variables in the model. This will be our next task. 

\subsubsection{Determining the Lagrange multiplier}
By operating with $\zh\times$ on the momentum equation in \eqref{thermal-l1} and using the definition of $\bu_{TG}$ in \eqref{trsw-tg}, we obtain
\begin{equation}
-\frac{\partial}{\partial t}\nabla\left(\frac{1}{\epsilon}\eta+\frac{1}{2} b\right) - \eta q \bu = \zh\times\nabla\left(\frac{1}{\epsilon}(1+\epsilon b)\zeta + \frac{1}{2}|\bu_{TG}|^2 + \frac{1}{\epsilon} B\right) - \left(\frac{1}{2\epsilon}(\epsilon\zeta-h)+\frac{1}{2\eta}B\right)\zh\times\nabla b\,.
\label{eq:rearranged-l1-q1}
\end{equation}
The first term above follows from using the definition of $\bu_{TG}$, by noting that the bathymetry has no time derivative. This allows us to rewrite $\bu_{TG}$ in terms of $\eta$ and $b$. Taking the time derivative through the gradient in \eqref{eq:rearranged-l1-q1} allows us to use the continuity equation and the advection equation to obtain an elliptic equation. Substituting the definition for $B$ in \eqref{B-def} then leads to the following equation which determines the Lagrange multiplier $\bu$,
\begin{equation}
\begin{aligned}
\frac{1}{\epsilon}\nabla(\nabla\cdot\eta\bu) + \frac{1}{2}\nabla(\bu\cdot\nabla b) &+ \zh\times \nabla\Big(\frac{1}{\epsilon}(1+\epsilon b)\zeta + \frac{1}{2}|\bu_{TG}|^2 + \frac{1}{\epsilon} \zh\cdot\nabla\times\Big(\eta(\bu-\bu_{TG})\Big)\Big)
\\&\quad
- \Big(\frac{1}{2\epsilon}(\epsilon \zeta-h)+\frac{1}{2\eta}\zh\cdot\nabla\times\Big(\eta(\bu-\bu_{TG})\Big)\Big)\zh\times\nabla b - \eta q\bu = 0\,.
\end{aligned}
\label{l1-diagnostic}
\end{equation}
Before going to the general case, let us consider the case in which the horizontal gradient of buoyancy \emph{vanishes}. 
\paragraph{No horizontal buoyancy gradients.} In this case, equation \eqref{l1-diagnostic} reduces to
\begin{equation}
\frac{1}{\epsilon}\nabla(\nabla\cdot\eta\bu) + \zh\times\nabla\Big(\frac{1}{\epsilon}\zeta + \frac{1}{2}|\bu_{TG}|^2 + \frac{1}{\epsilon}\zh\cdot\nabla\times\Big(\eta(\bu-\bu_{TG})\Big)\Big) -  \eta q\bu = 0\,.
\label{eq:AH}
\end{equation}
This is the diagnostic partial differential equation used to determine $\bu$ in \cite{salmon1983practical} when variations in bathymetry are absent and it is identical to equation (3.16) in \cite{allen1996extended}. After applying the identity 
\begin{equation}
\Delta \bu = \nabla^\perp(\nabla^\perp\cdot\bu) + \nabla(\nabla\cdot\bu)
\label{eq:Laplace-identity}
\end{equation}
in equation \eqref{eq:AH}, we can rewrite the diagnostic equation \eqref{l1-diagnostic} in simpler form. Here we have used the perpendicular "$\perp$" notation for brevity, see remark \ref{rem:perpnotation}. The Laplacian identity \eqref{eq:Laplace-identity} implies that the equation which determines $\bu$ is a linear non-autonomous elliptic partial differential equation (PDE), given by
\begin{equation}
\frac{1}{\epsilon}\Delta(\eta\bu) -  \eta q \bu = -\nabla^\perp\left(\frac{1}{\epsilon}\zeta + \frac{1}{2}|\bu_{TG}|^2 - \frac{1}{\epsilon}\nabla^\perp\cdot\eta\bu_{TG}\right) .
\label{eq:inhomhelmholtz}
\end{equation}
Note that the coefficient $\eta q$ in \eqref{eq:inhomhelmholtz} is the \emph{total vorticity}, since $\eta q \bu = (\frac{1}{\epsilon} +f_1+\Delta\zeta)\bu$. The solution behaviour of the elliptic equation \eqref{eq:inhomhelmholtz} for the quantity $\eta\bu$ depends on the sign of the potential vorticity, $q$, in the following three cases
\begin{enumerate}
\item $q >0$. The equation for $\eta \bu$ is a screened Poisson equation.
\item $q <0$. The equation for $\eta \bu$ is an inhomogeneous Helmholtz equation.
\item $q =0$. The equation for $\eta \bu$ is a Poisson equation.
\end{enumerate}
In the situation being considered at the moment, there are no horizontal buoyancy gradients. Consequently, the potential vorticity $q$ is preserved along Lagrangian particle trajectories and does not change sign during the calculation. This means that such changes in the solution behaviour of \eqref{eq:inhomhelmholtz} do not occur. Thus, the `equator', where $f$ changes sign, acts as a boundary between the `northern and southern hemispheres', in the absence of horizontal buoyancy gradients. In this case, Lagrangian particles which start in the northern hemisphere stay in the northern hemisphere, because their potential vorticity is conserved in the absence of horizontal buoyancy gradients. The case with non-zero horizontal buoyancy gradients is the general case, which we will discuss now.

\paragraph{General case.} 
When horizontal gradients of buoyancy are nonzero, equation \eqref{eq:inhomhelmholtz} for the determination of the Lagrange multiplier $\bu$ becomes considerably more extensive
\begin{equation}
\begin{aligned}
\frac{1}{\epsilon}\Delta(\eta\bu) &+ \frac{1}{2}\nabla(\bu\cdot\nabla b)
- \frac{1}{2}(\nabla^\perp\cdot\bu)\nabla^\perp b -\frac{1}{2\eta} (\bu\cdot\nabla^\perp\eta)\nabla^\perp b
- \eta q \bu 
=  
- \frac{1}{\epsilon}\nabla^\perp \big((1+\epsilon b)\zeta\big)
+ \frac{1}{2\epsilon}(\epsilon\zeta-h)\nabla^\perp b
\\&\quad
- \frac{1}{2}\nabla^\perp|\bu_{TG}|^2- \frac{1}{2\eta}\left(\nabla^\perp\cdot\eta\bu_{TG}\right)\nabla^\perp b  + \frac{1}{\epsilon}\nabla^\perp\nabla^\perp\cdot\Big(\eta\bu_{TG}\Big) .
\end{aligned}
\label{TL1-soln}
\end{equation}
The coefficient for the $\bu$ term now has an additional contribution from the perpendicular gradient of the buoyancy, which changes the conditions for the type of PDE. The zeroth order terms in the presence of horizontal buoyancy gradients indicate that the equator is no longer a stationary boundary between the northern and southern hemispheres. Indeed, the perpendicular gradients of buoyancy in combination with perpendicular gradients of the depth have removed the `equatorial boundary'. Likewise, when the horizontal gradient of buoyancy is included, the potential vorticity is not conserved along Lagrangian particle trajectories. Moreover, the effects of the first order terms at this point remain unexamined. At this point, we shall defer further discussion of these elliptic equations until section \ref{sec:TQGmodel} and leave the discussion of the interpretation of the effects of horizontal buoyancy gradients on the solution behaviour of the elliptic equation \eqref{TL1-soln} for the quantity $\eta\bu$ for the TL1 model as an open problem. 
\medskip

\noindent Before substituting the asymptotic expansions and truncating the equations to achieve thermal geostrophic balance in terms of strict asymptotics in section \ref{sec:TQGmodel}, we will first derive the \emph{stochastic} thermal L1 equations. 

\section{The Eliassen approximation of stochastic TRSW} \label{sec:stochEliassen}
The equation sets for the deterministic and stochastic TRSW models in section \ref{DTRSW-sec} and the deterministic TL1 model in the previous section have all been derived in the variational framework of the Euler--Poincar\'e theorem introduced in section \ref{sec-EPframework}. The corresponding Kelvin circulation laws for each of these theories follows from their Kelvin--Noether theorem \ref{thm:Kelvin}, proved in section \ref{sec-KNthm}. Let us now derive the stochastic version of the TL1 equations and their corresponding Kelvin circulation law by following the SALT formulation in the Euler--Poincar\'e variational framework. To do so, we first investigate the balance relation in the presence of stochasticity.

\subsection{Stochastic thermal geostrophic balance}
To obtain the stochastic thermal geostrophic balance, we start from the TRSW equations with SALT, given in \eqref{eq:STRSW} by
\begin{equation}
\begin{aligned}
{\sf d}\bu + ({\sf d}\boldsymbol \chi_t\cdot\nabla)\bu  + \sum_k(\nabla \boldsymbol \xi_k)\cdot\bu\circ dW_t^k 
&= - \frac{\alpha}{{\rm Fr}^2}\nabla\big((1+\mathfrak{s}b) \zeta\big)\,dt + \frac{\mathfrak{s}}{2\,{\rm Fr}^2}(\alpha\zeta-h)\nabla b\,dt 
\\&\qquad 
- \frac{1}{{\rm Ro}}f\hat{\mathbf{z}}\times {\sf d}\boldsymbol \chi_t - \frac{1}{{\rm Ro}}\sum_k\nabla(\boldsymbol \xi_k\cdot\mathbf{R})\circ dW_t^k
\,,\\
{\sf d}\eta + \nabla\cdot(\eta\, {\sf d}\boldsymbol \chi_t) &= 0
\,,\\
{\sf d} b + ({\sf d}\boldsymbol \chi_t\cdot\nabla)b &= 0\,
\end{aligned}
\label{eq:STRSW-tgb}
\end{equation}
with boundary conditions in \eqref{eq:bdySTRW}. We recall the assumptions that led to deterministic thermal geostrophic balance. As before, we assume that the magnitudes of the dimensionless numbers in the problem is $\mathcal O(\alpha) = \mathcal O(\mathfrak{s}) = \mathcal O({\rm Ro}) = \mathcal O({\rm Fr})$. This asymptotic regime allows us to continue with a single small parameter, $\epsilon$. So, we formulate \eqref{qg-assumptions} as
\begin{equation}
\begin{aligned}
f(\bx) &= 1 + \epsilon f_1(\bx),\\
h(\bx) &= 1 + \epsilon h_1(\bx),\\
\bR(\bx) &= \bR_0(\bx)+\epsilon \bR_1(\bx),
\end{aligned}
\label{qg-assumptions2}
\end{equation}
where the additional relations $\zh\cdot\nabla\times\bR_0=1$ and $\zh\cdot\nabla\times\bR_1=f_1$ hold. Upon substituting the asymptotic expansions \eqref{qg-assumptions2} into the stochastic TRSW equations \eqref{eq:STRSW-tgb} and collecting all terms of $\mathcal O(\epsilon^{-1})$, we find
\begin{equation}
\left(\nabla \zeta + \frac{1}{2}\nabla b  + \zh\times\bu \right)\,dt + \sum_k\big(\zh\times\boldsymbol \xi_k + \nabla(\boldsymbol \xi_k\cdot\bR_0)\big)\circ dW_t^k = 0\,.
\end{equation}
The drift part of this stochastic partial differential equation is the deterministic thermal geostrophic balance \eqref{trsw-bal} and the diffusion part provides us with a relation between the noise amplitude $\boldsymbol \xi_k$ and the vector potential for the Coriolis parameter,
\begin{equation}
\begin{aligned}
\bu &= \zh\times\nabla\left(\zeta + \frac{1}{2}b\right) + \mathcal O(\epsilon),\\
\boldsymbol \xi_k &= \zh\times\nabla(\boldsymbol \xi_k\cdot\bR_0)+ \mathcal O(\epsilon)\,.
\end{aligned}
\label{eq:stochasticgeostrophicbalance}
\end{equation}
Since the Brownian motions are assumed to be independent, \eqref{eq:stochasticgeostrophicbalance} needs to be satisfied for each $k$. We can identify the thermal geostrophic balance velocity field as $\bu_{TG} = \zh\times\nabla(\zeta + \frac{1}{2}b)$ and expand the velocity field $\bu$ as in the deterministic case
\begin{equation}
\bu = \bu_{TG} + \mathcal O(\epsilon),
\end{equation}
and following the same reasoning, the $\boldsymbol \xi_k$ can be expanded as
\begin{equation}
\boldsymbol \xi_k = \boldsymbol \xi_{TG_k} + \mathcal O(\epsilon).
\end{equation}
We can now investigate the stochastic thermal L1 model.

\subsection{Stochastic TL1}
The equations governing the stochastic TL1 model are obtained in the SALT formulation by applying the stochastic Euler--Poincar\'e theorem \ref{thm:SEP} to the TL1 Lagrangian, given by \eqref{trsw-lagAH4}. This incorporates the deterministic geostrophic balance. We find the stochastic version of the TL1 equations \eqref{thermal-l1}, given by
\begin{equation}
\begin{aligned}
{\sf d} \bu_{TG} - {\sf d}\boldsymbol \chi_t &\times \left( \nabla\times\Big(\bu_{TG}+\frac{1}{\epsilon}\bR\Big)\right)
+ \nabla\left(\frac{1}{\epsilon}(1+\epsilon b)\zeta + \frac{1}{2}|\bu_{TG}|^2 + \frac{1}{\epsilon} B\right)\,dt
\\&\quad
+ \nabla\left(\sum_k \boldsymbol \xi_k\cdot \Big(\bu_{TG}+\frac{1}{\epsilon}\bR\Big)\right)\circ dW_t^k- \left(\frac{1}{2\epsilon}(\epsilon\zeta-h)+\frac{1}{2\eta}B\right)\nabla b\,dt = 0\,.
\\
{\sf d} \eta + \nabla\cdot\big(\eta\,{\sf d}\sym{\chi}_t\big) &= 0,\\
{\sf d} b + {\sf d}\sym{\chi}_t\cdot\nabla b &= 0\,.
\end{aligned}
\label{Stoch-thermal-l1}
\end{equation}
Here, the function $B$ is defined by
\begin{equation}
B := \zh\cdot\nabla\times\Big(\eta\big(\bu-\bu_{TG}\big)\Big)\,.
\label{dB-def}
\end{equation}
The boundary conditions are given by
\begin{equation}
\mathbf{\hat{n}}\cd \bu=0
\quad\hbox{and}\quad
\mathbf{\hat{n}}\cd \boldsymbol{\xi}_k = 0
\quad\hbox{and}\quad
\mathbf{\hat{n}}\times\nabla b = 0
\quad\hbox{on the boundary $\partial{\cal D}$.}
\label{eq:bdySTL1}
\end{equation} 
The Kelvin--Noether theorem \ref{thm:Kelvin} for the stochastic TL1 model is given by the following theorem.
\medskip

\begin{theorem}[Kelvin theorem for the stochastic TL1 model]
The stochastic TL1 equations \eqref{Stoch-thermal-l1} imply the following Kelvin circulation law
\begin{equation}
\begin{aligned}
{\sf d}\oint_{c({\sf d}\sym{\chi}_t)} \left(\bu_{TG} + \frac{1}{\epsilon}\bR\right)\cdot d\bx 
&= \oint_{c({\sf d}\sym{\chi}_t)} \left( \frac{1}{2\epsilon}(\epsilon\zeta-h) + \frac{1}{2\eta} B\right)\nabla b\cdot d\bx\,dt
\\
&= \int\!\!\int_{\partial S = c({\sf d}\sym{\chi}_t)} \zh\cdot\nabla\left( \frac{1}{2\epsilon}(\epsilon\zeta-h)
+\frac{1}{2\eta} B\right)\times\nabla b \, dx\,dy\,dt\,.
\end{aligned}
\label{StochL1-kelvin}
\end{equation} 
\end{theorem}

\begin{proof}
The proof follows the pattern of the standard Kelvin--Noether theorem \ref{thm:Kelvin}, modulo an application of the Kunita--It\^o--Wentzell theorem which provides the chain rule for the Lie derivatives of differential forms by stochastic vector fields, as proved in \cite{de2020implications}. 
The loop $c({\sf d}\sym{\chi}_t)$ does not explicitly require the evaluation of a stochastic integral, as it is the push-forward of a stationary loop by the flow which is generated by the vector field ${\sf d}\sym{\chi}_t$, see remark \ref{rem:kelvin}. 
\end{proof}
\medskip

\noindent The stochastic TL1 equations do not conserve energy due to their explicit dependence on time via the noise. From the Kelvin circulation theorem associated to the stochastic TL1 equations, an evolution equation for potential vorticity can be derived. This equation shows that potential vorticity is not conserved along Lagrangian particle trajectories, but is generated by the effect also present on the right hand side in the Kelvin circulation theorem \ref{StochL1-kelvin}. Recall that the potential vorticity $q$ is defined by 
\begin{equation}
q = \frac{1}{\eta}\left(\omega + \frac{1}{\epsilon}f\right) = \frac{1}{\eta}\left(\Delta\Big(\zeta + \frac{1}{2}b\Big)+\frac{1}{\epsilon} + f_1\right)
. 
\label{def:pvSTL1}
\end{equation} 
The evolution equation for $q$ is given by
\begin{equation}
{\sf d}q + {\sf d}\boldsymbol \chi_t\cdot\nabla q = \frac{1}{\eta}\zh\cdot\nabla\left( \frac{1}{2\epsilon}(\epsilon\zeta-h)
+\frac{1}{2\eta}B\right)\times\nabla b\,dt\,.
\end{equation}
Even though potential vorticity is not a Lagrangian invariant, the stochastic TL1 equations have an infinite family of integral conservation laws, given by
\begin{equation}
C_{\Phi,\Psi} = \int_{\cal D} \eta\Phi(b) + \eta q\Psi(b) \,dx\,dy\,.
\end{equation}
The proof for these conservation laws is a direct calculation that uses the boundary conditions \eqref{eq:bdySTL1}. In order to use \eqref{Stoch-thermal-l1} as a predictive model, one must be able to determine $\bu$ from the other variables in the model. We can proceed as in the deterministic case and derive an elliptic equation for $\bu$.

\subsubsection{Determining the Lagrange multiplier}
By operating with $\zh\times$ on the momentum equation in \eqref{Stoch-thermal-l1} and using the definition of $\bu_{TG}$ \eqref{trsw-tg}, we obtain
\begin{equation}
\begin{aligned}
{\sf d}\nabla\left(\frac{1}{\epsilon}\eta + \frac{1}{2}b\right) - \eta q\, {\sf d}\boldsymbol \chi_t &= \zh\times\nabla\left(\frac{1}{\epsilon}(1+\epsilon b)\zeta + \frac{1}{2}|\bu_{TG}|^2 + \frac{1}{\epsilon} B\right)dt + \nabla\left(\sum_k \boldsymbol \xi_k\cdot \Big(\bu_{TG}+\frac{1}{\epsilon}\bR\Big)\right)\circ dW_t^k
\\&\quad
- \left(\frac{1}{2\epsilon}(\epsilon\zeta-h)+\frac{1}{2\eta}B\right)\zh\times\nabla b\,dt\,.
\end{aligned}
\end{equation}
We continue by taking the stochastic differential through the gradient and substitute the continuity equation and the advection equation for the buoyancy from \eqref{Stoch-thermal-l1}. By using the definition for $B$ in \eqref{dB-def}, we can then use the vector calculus identity $\Delta \bu = \nabla^\perp\nabla^\perp\cdot \bu + \nabla\nabla\cdot \bu$. This leads to two linear non-autonomous elliptic partial differential equations, one for the drift part and one for the diffusion part. The elliptic equation for the drift part is given by
\begin{equation}
\begin{aligned}
\frac{1}{\epsilon}\Delta(\eta\bu) &+ \frac{1}{2}\nabla(\bu\cdot\nabla b)
- \frac{1}{2}(\nabla^\perp\cdot\bu)\nabla^\perp b -\frac{1}{2\eta} (\bu\cdot\nabla^\perp\eta)\nabla^\perp b
- \eta q \bu 
=  
- \frac{1}{\epsilon}\nabla^\perp \big((1+\epsilon b)\zeta\big)
+ \frac{1}{2\epsilon}(\epsilon\zeta-h)\nabla^\perp b
\\&\quad
- \frac{1}{2}\nabla^\perp|\bu_{TG}|^2- \frac{1}{2\eta}\left(\nabla^\perp\cdot\eta\bu_{TG}\right)\nabla^\perp b  + \frac{1}{\epsilon}\nabla^\perp\nabla^\perp\cdot\Big(\eta\bu_{TG}\Big),
\end{aligned}
\label{eq:stl1-drift}
\end{equation}
and the elliptic equation for the diffusion part is given by
\begin{equation}
\frac{1}{\epsilon}\nabla(\nabla\cdot\eta\boldsymbol \xi_k) + \frac{1}{2}\nabla(\boldsymbol\xi_k\cdot\nabla b) + \nabla^\perp\left(\boldsymbol \xi_k\cdot\Big(\bu_{TG}+\frac{1}{\epsilon}\bR\Big)\right) - \eta q\boldsymbol \xi_k = 0.
\label{eq:stl1-diffusion}
\end{equation}
These elliptic equations at leading order provide the deterministic and the stochastic geostrophic balance conditions. 

We will not pursue the properties of the stochastic TL1 equations any further here. 
Instead, we will apply further asymptotic analysis to derive the thermal quasi-geostrophic (TQG) equations from the TL1 model, investigate the properties of their deterministic solutions and then derive the corresponding stochastic TQG equations by using the SALT approach.
%We will further investigate these elliptic equations in appendix \ref{appendix:TQG}, in the context of thermal quasi-geostrophy (TQG), which is obtained upon introducing the asymptotic expansions \eqref{qg-assumptions2} and truncating at $\mathcal O(1)$. We will visit the deterministic case first and then proceed to discuss the stochastic case.

\section{Thermal QG model}\label{sec:TQGmodel}
In sections \ref{sec:detEliassen} and \ref{sec:stochEliassen}, we made an approximation to the kinetic energy by assuming that the velocity field can be decomposed into a thermal geostrophic part and a higher order part. This led to the thermal L1 (TL1) equations, given by \eqref{thermal-l1} in the deterministic case, and by \eqref{Stoch-thermal-l1} in the stochastic case. These TL1 equations, in turn, can be approximated further to yield the motion equations for thermal quasi-geostrophy (TQG), which will be the subject of this section. We will start with the deterministic TL1 case, continuing from where we stopped in section \ref{sec:detEliassen} and discuss the solution properties of a numerical example of the TQG equations. We will finish this section by looking into the Hamiltonian formulation of the deterministic TQG equations. We will then use the Hamiltonian framework to derive the stochastic TQG equations.

\subsection{Deterministic TQG model}\label{sec:detTQG}
To obtain the thermal quasi-geostrophic model, one could expand the TL1 equations \eqref{thermal-l1} and find that the continuity equation features the divergence of $\bu$, which can then be solved for by substitution. This approach has the disadvantage that equations are expanded before substitution, which loses accuracy. Instead we follow the derivation for the elliptic equation which determines the Lagrange multiplier $\bu$. By operating with $\zh\times$ on the TL1 momentum equation in \eqref{thermal-l1} and using the definition of $\bu_{TG}$, we obtain
\begin{equation}
\frac{\partial}{\partial t}\nabla\left(\frac{1}{\epsilon}\eta+\frac{1}{2} b\right) + \eta q \bu = \zh\times \left(\nabla\Big(\frac{1}{\epsilon}(1+\epsilon b)\zeta + \frac{1}{2}|\bu_{TG}|^2 + \frac{1}{\epsilon} B\Big) - \Big(\frac{1}{2\epsilon}(\epsilon\zeta-h)+\frac{1}{2\eta}B\Big)\nabla b\right).
\label{eq:rearranged-l1-q12}
\end{equation} 
We now take the divergence of \eqref{eq:rearranged-l1-q12} and substitute the continuity equation for the depth, which yields
\begin{equation}
\begin{aligned}
%\frac{\partial}{\partial t}\Delta\Big(\zeta_1+\frac{1}{2}b_1\Big) + q\nabla \cdot\eta\bu + \eta\bu\cdot\nabla q &= J\left(\frac{1}{2\epsilon^2}(\zeta-h) + \frac{1}{2\epsilon\eta}B,b\right)
%\,,\\
\frac{\partial}{\partial t}\Delta\Big(\frac{1}{\epsilon}\eta+\frac{1}{2}b\Big)-q\frac{\partial}{\partial t}\eta + \eta\bu\cdot\nabla q &= J\left(\frac{1}{2\epsilon}(\epsilon\zeta-h) + \frac{1}{2\eta}B,b\right)
,\end{aligned}
\label{eq:preTQG}
\end{equation}
so the potential energy term $\partial \eta/\partial t$ also appears in the vorticity equation. 
Equivalently, one can take the two dimensional curl, or the perpendicular divergence, to arrive directly at \eqref{eq:preTQG}. In a moment, we will expand these equations in the asymptotic regime introduced in \eqref{qg-assumptions} and truncate at $\mathcal O(1)$ to obtain the thermal quasi-geostrophic equations (TQG). In the asymptotic expansions to follow, it will be helpful to note that the potential vorticity, defined in \eqref{def:pvtL1}, contains a term that is of order $\mathcal O(\epsilon^{-1})$, since it allows for the substitution of $\zeta$.

\paragraph{Asymptotic expansion to the deterministic TQG regime.}
By means of the asymptotic expansions \eqref{qg-assumptions} which were used to derive an expression for $\bu_{TG}$, one can expand equation \eqref{eq:preTQG} and collect terms of the same order. Truncating at $\mathcal O(1)$ then leads to the following motion equation for thermal quasi-geostrophy (TQG) \cite{warneford2013quasi, zeitlin2018geophysical}
\begin{equation}
\frac{\partial}{\partial t}\left(\Delta\Big(\zeta+\frac{1}{2}b\Big)-\zeta\right) + \bu_{TG}\cdot\nabla\left(\Delta\Big(\zeta+\frac{1}{2}b\Big)+f_1 \right)  = \frac12 \zh\times \nabla (\zeta-h_1)\cdot \nabla b
\,.
\label{eq:PVdyn-TQG1}
\end{equation}
Because the depth equation in \eqref{thermal-l1} was substituted into equation \eqref{eq:preTQG}, the potential energy term $\partial \zeta/\partial t$ remains in the vorticity equation \eqref{eq:PVdyn-TQG1}. In the asymptotic expansion, the deterministic buoyancy equation keeps its form, as
\begin{equation}
\frac{\partial}{\partial t}b + \bu_{TG}\cdot\nabla b = 0,
\label{eq:bTQG}
\end{equation}
and the boundary conditions at this order become
\begin{equation} 
\mathbf{\hat{n}}\cd \bu_{TG}=0
\quad\hbox{and}\quad
\mathbf{\hat{n}}\times\nabla b = 0
\quad\hbox{on the boundary $\partial{\cal D}$.}
\label{eq:bdyTQG}
\end{equation}
Equations \eqref{eq:PVdyn-TQG1} and \eqref{eq:bTQG} together with the boundary conditions \eqref{eq:bdyTQG} form a closed model which approximates the TL1 model \eqref{thermal-l1}. This completes the present derivation of the thermal quasi-geostrophic (TQG) model, cf., \cite{warneford2013quasi, zeitlin2018geophysical}.

The vorticity equation \eqref{eq:PVdyn-TQG1} for TQG can be rewritten in terms of the stream function $\psi:=\zeta+\frac{1}{2}b$ and the Jacobian operator $J(\psi,a):=\zh\cdot \nabla \psi \times \nabla a = \bu_{TG}\cdot\nabla a$ for a function $a(\bx)$ to obtain the more compact form,
\begin{equation}
\frac{\partial}{\partial t}\left(\Delta\psi-\psi+b\right) + J(\psi,\Delta\psi+f_1) = -\,\frac12 J(h_1, b).
\label{eq:qTQG}
\end{equation}
Here, we have added and subtracted $\frac{1}{2}b$ in the time derivative in vorticity equation \eqref{eq:PVdyn-TQG1}. Another $\frac{1}{2}b$ contribution comes from the forcing term on the right hand side, upon replacing the free surface elevation $\zeta$ by the stream function $\psi$ and then using the buoyancy equation \eqref{eq:bTQG}, written now as
\begin{equation}
\frac{\partial}{\partial t}b + J(\psi,b) = 0.
\label{eq:bTQGstream}
\end{equation} 
The boundary conditions are
\begin{equation} 
\mathbf{\hat{n}}\times \nabla\psi=0
\quad\hbox{and}\quad
\mathbf{\hat{n}}\times\nabla b = 0
\quad\hbox{on the boundary $\partial{\cal D}$.}
\label{eq:bdyTQGstream}
\end{equation}
When the bathymetry is flat and the Coriolis parameter is constant, equation \eqref{eq:qTQG} reduces to the TQG equation found in \cite{warneford2013quasi, zeitlin2018geophysical}. Since the Jacobian operator is zero when the arguments are functionally related, it is possible to write the deterministic TQG equation in terms of a type of potential vorticity variable, which we will call $q$. This notation forms a close link between QG without buoyancy and TQG, with
\begin{equation}
q := \Delta \psi - \psi + f_1.
\label{def:qfortqg}
\end{equation}
The definition of $q$ in \eqref{def:qfortqg} allows us to formulate the vorticity equation in \eqref{eq:qTQG} as
\begin{equation}
\frac{\partial}{\partial t}(q+b) + J(\psi, q) = -\frac{1}{2}J(h_1,b),
\label{eq:TQGinqform}
\end{equation}
The formulation in terms of $q$ as in \eqref{eq:TQGinqform} is particularly useful in showing that these equations conserve energy, but also to note that the TQG equations can be related to Rayleigh-B\'enard convection. 
\medskip

\begin{remark}[Rayleigh-B\'enard convection]
One can write \eqref{eq:TQGinqform} in such a way that all terms that depend on the buoyancy variations appear on the right hand side,
\begin{equation}
\begin{aligned}
\frac{\partial}{\partial t}q + J(\psi, q) &= \frac{1}{2}J(\psi, b)+\frac{1}{2}J(\zeta-h_1,b),\\
\frac{\partial}{\partial t}b + J(\psi, b) &= 0.
\end{aligned}
\label{eq:TQGinqform2}
\end{equation}
This notation reveals that the equations \eqref{eq:TQGinqform2} are reminiscent of the ideal Rayleigh-B\'enard convection equations. The Rayleigh-B\'enard convection problem in a vertical $xz$--plane, formulated in terms of vorticity and stream function, is given by
\begin{equation}
\begin{aligned}
\frac{\partial}{\partial t} \widetilde{q} + J(\psi,\widetilde{q}) &= \alpha g T_z,\\
\frac{\partial}{\partial t} T + J(\psi,T) &= 0.
\end{aligned}
\end{equation}
Here $\widetilde{q}=\Delta\psi$ is the vorticity, $T$ is the temperature, $\psi$ is the stream function, $g$ is gravity and $\alpha$ is the thermal expansion coefficient. For the typical Rayleigh-B\'enard convection problem, in the vertical direction there are two solid boundaries and in the horizontal direction, one either uses periodic boundary conditions or solid boundaries. The bottom boundary is heated and the top boundary is cooled, in such a way that the temperature difference is constant. Similar boundary conditions can be established for the thermal quasi-geostrophic equations. The main differences between the two models is that the forcing terms in TQG involve derivatives in every direction, whereas Rayleigh-B\'enard convection only involves derivatives of the temperature in the vertical $z$-direction. Also, TQG is obtained as a model based on thermal geostrophic balance, whereas the Rayleigh-B\'enard model does not impose any balance relation.
\end{remark}

\paragraph{Elliptic equation for TL1.}
Upon substituting the asymptotic expansions \eqref{qg-assumptions} and collecting the leading order terms, the elliptic equation for TL1 \eqref{TL1-soln} reads
\begin{equation}
\begin{aligned}
\frac{1}{\epsilon}\Delta\bu_{TG} - \frac{1}{\epsilon}\bu_{TG} 
=  
- \frac{1}{\epsilon}\nabla^\perp\zeta
+ \frac{1}{2\epsilon}\nabla^\perp b
+ \frac{1}{\epsilon}\nabla^\perp\nabla^\perp\cdot\bu_{TG} + \mathcal O(1) .
\end{aligned}
\label{eq:ellipticTQG}
\end{equation}
By means of the identity $\Delta\bu_{TG} = \nabla\nabla\cdot\bu_{TG} + \nabla^\perp\nabla^\perp\cdot\bu_{TG}$ and using the fact that $\bu_{TG}$ is  divergence-free, we obtain at leading order the definition for $\bu_{TG}=\nabla^\perp(\zeta+\frac12 b)$. At the next order, the elliptic equation will provide an expression for $\bu$ that is consistent with the asymptotic regime. 

\subsection{Numerical TQG example}\label{sec:TQG-numerical example}

We implemented the TQG equations \eqref{eq:TQGinqform2} using finite element methods (FEM) for the spatial variables. The FEM algorithm we used is an adaptation of the algorithm formulated in \cite{bernsen2006femdis}, and was implemented using
Firedrake\footnote{http://www.firedrakeproject.org/index.html}, see \cite{Rathgeber2017}.
In particular, we approximate the vorticity and buoyancy fields in first order discrete Galerkin finite element space, and approximate the stream function in first order continuous Galerkin finite element space.
For the time step, we used an optimal third order strong stability preserving Runge-Kutta method, see \cite{gottlieb2005high,cotter2019numerically}.

Figure \ref{Fig:CourtesyofWeiPan2} shows a snapshot at a certain time taken from
a high resolution numerical run of the TQG equations. In this numerical example, we used the following boundary and initial conditions. The domain is $[0, 2\pi]^2$ discretised at a resolution of $512^2$. The boundary conditions are periodic in the vertical direction and walls in the horizontal direction. The parameters and initial conditions are
\begin{equation}
\begin{aligned}
h_1(x,y) &= \cos x + \frac{1}{2}\cos 2x + \frac{1}{3}\cos 3x\,,\\
f_1(x,y) &= 0\,,\quad\hbox{($f$-plane)}\\
b(x,y,0) &= -\frac{1}{1+\exp(-x+\pi)}+\frac{1}{2}\,,\\
q(x,y,0) &= -\exp(-5(y-\pi)^2)\,,\\
\zeta(x,y,0) &= \psi(x,y,0) - \frac{1}{2}b_1(x,y,0)\,.
\end{aligned}
\end{equation}
The stream function $\psi=\zeta + b/2$ is calculated from the potential vorticity $q$ by means of the elliptic problem given in \eqref{def:qfortqg}.

\begin{figure}[H]
\begin{minipage}{0.5\textwidth}
\centering
\includegraphics[width=.95\textwidth]{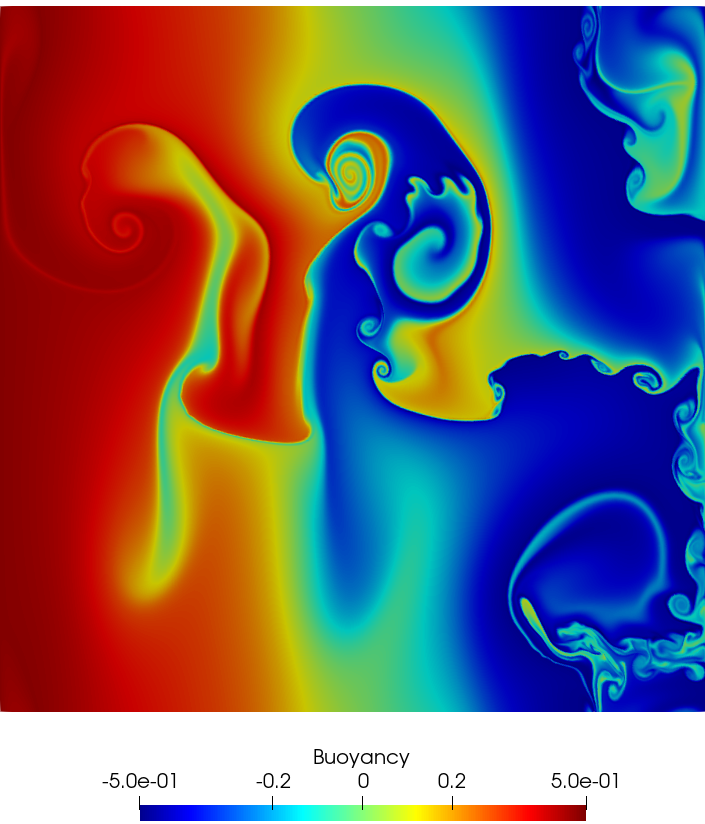}

\vspace{1cm}

\includegraphics[width=.95\textwidth]{ssh1.png}
\end{minipage}
\begin{minipage}{0.5\textwidth}
\centering
\includegraphics[width=.95\textwidth]{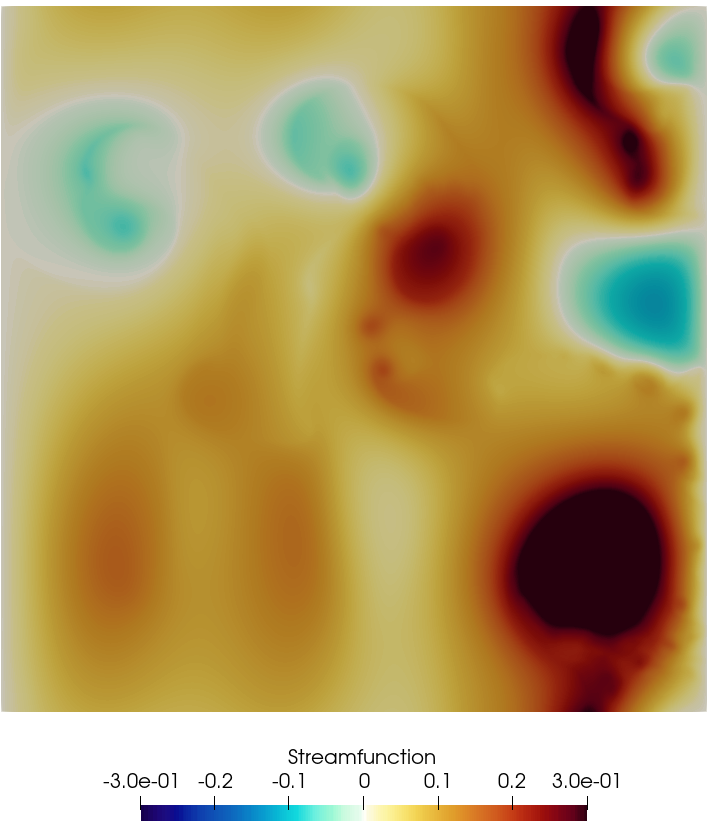}

\vspace{1cm}

\includegraphics[width=.95\textwidth]{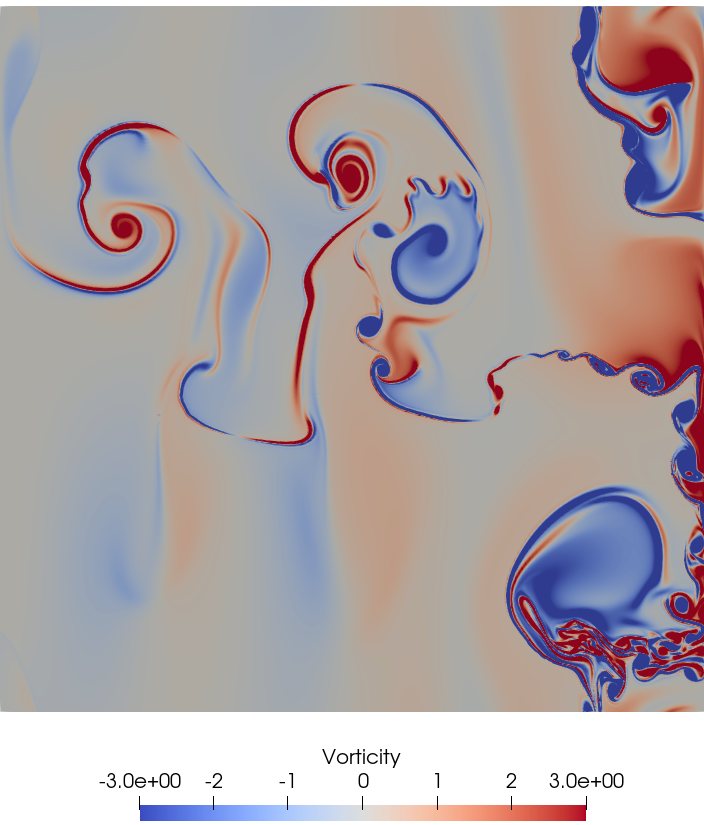}
\end{minipage}
\caption{\footnotesize These images from an evolutionary computational simulation of deterministic TQG illustrate the result of the creation of a buoyancy front from an initial state and its subsequent emergent instabilities which produce a cascade of horizontal circulations to smaller scales. In the top left, the buoyancy profile is shown and in the bottom left, the sea surface height is shown. In the top right, the stream function is shown and in the bottom right, the vorticity is shown. The domain is periodic at the upper and lower boundaries, while the stream functions are  constant and equal on the lateral boundaries. The bathymetry varies in the lateral direction as a sum of cosines with wavenumbers $k=1, 2, \hbox{ and } 3$. The initial condition had a Gaussian-profile strip of vorticity along the lateral mid-line and the buoyancy began with a $k = 1$ sine profile in the lateral direction. The figure shows that a buoyancy front had developed and then generated a cascade of smaller mushroom-like dipole circulations and trains of Kelvin-Helmholtz roll-ups via interaction between the buoyancy and bathymetry gradients, followed by shear instability.  Compare this configuration with the chlorophyll tracers shown in Figure \ref{Fig:globcurrent1}. }\label{Fig:CourtesyofWeiPan2}
\end{figure}

\subsection{Hamiltonian formulation of TQG}\label{sec:hamTQG}
In this section we will investigate the Hamiltonian structure of TQG. The Hamiltonian formulation of baroclinic QG and QG was analysed in \cite{holm1986hamiltonian,holm1998hamilton}, respectively.
\paragraph{Conservation laws for deterministic TQG.} The deterministic TQG system \eqref{eq:qTQG} and \eqref{eq:bTQGstream} conserves the energy
\begin{equation}
H_{TQG}(q,b) =  -\,\frac12 \int_{\cal D} (q-f_1)(\Delta-1)^{-1}(q-f_1) + h_1  b\Big)\,dx\,dy\,.
\label{eq:erg-tqg}
\end{equation}
%\begin{proof}
%\begin{equation}
%\begin{aligned}
%-\,\frac{\partial}{\partial t}H_{TQG} 
%&= 
% \frac12\frac{\partial}{\partial t}\int_{\cal D} (q-f_1)(\Delta-1)^{-1}(q-f_1) + h_1  b_1\Big)\,dx\,dy \\
%%&= 
%%\frac{1}{2}\int_{\cal D} \frac{\partial}{\partial t}q(\Delta-1)^{-1}(q-f_1) + (q-f_1)(\Delta-1)^{-1}\frac{\partial}{\partial t} q + h_1\frac{\partial}{\partial t}b_1\,dx\,dy \\
%&=
%\int_{\cal D} \psi \frac{\partial}{\partial t}q + \frac{1}{2}h_1\frac{\partial}{\partial t} b_1 \,dx\,dy\\
%&=
%\int_{\cal D} -\psi \frac{\partial}{\partial t}b_1 - \psi J(\psi, q) - \frac{1}{2}\psi J(h_1,b_1) - \frac{1}{2}h_1 J(\psi, b_1)\,dx\,dy\\
%&= 
%\int_{\cal D} \psi J(\psi,b_1) - \psi J(\psi, q) - \frac{1}{2}\psi J(h_1,b_1) - \frac{1}{2}h_1 J(\psi, b_1)\,dx\,dy\\
%&= 
%\int_{\cal D} -b_1 J(\psi,\psi) + q J(\psi,\psi) - \frac{1}{2}\psi J(h_1,b_1) + \frac{1}{2}\psi J(h_1, b_1)\,dx\,dy\\
%&= 0.
%\end{aligned}
%\end{equation}
%The integration by parts require the boundary conditions \eqref{eq:bdyTQGstream}.
%\end{proof}
The proof that the TQG equations conserve energy is a direct calculation that requires the boundary conditions \eqref{eq:bdyTQGstream}. 
The deterministic TQG equations also conserve an infinity of integral conservation laws, determined by two arbitrary differentiable functions of buoyancy $\Phi(b)$ and $\Psi(b)$ as
\begin{equation}
C_{\Phi,\Psi} = \int_{\cal D}\Phi(b) + q \Psi(b)\,dx\,dy\,.
\label{eq:casimirstqg}
\end{equation}
%\begin{proof}
%\begin{equation}
%\begin{aligned}
%\frac{\partial}{\partial t}C_{\Phi,\Psi} &= \frac{\partial}{\partial t} \int_{\cal D} \Phi(b_1) + q \Psi(b_1)\,dx\,dy \\
%&= 
%\int_{\cal D} \Phi'(b_1)\frac{\partial}{\partial t}b_1 + \frac{\partial}{\partial t}q \Psi(b_1) + q\Psi'(b_1)\frac{\partial}{\partial t}b_1 \,dx\,dy \\
%%&=
%%\int_{\cal D} -\Phi'(b_1)J(\psi,b_1) - \Psi(b_1) \frac{\partial}{\partial t}b_1 - \Psi(b_1) J(\psi, q) - \frac{1}{2}\Psi(b_1) J(h_1,b_1) - q\Psi'(b_1)J(\psi,b_1) \,dx\,dy \\
%&= 
%\int_{\cal D} -\Phi'(b_1)J(\psi,b_1) + \Psi(b_1)J(\psi,b_1) - \Psi(b_1) J(\psi, q) - \frac{1}{2}\Psi(b_1) J(h_1,b_1) - q\Psi'(b_1)J(\psi,b_1) \,dx\,dy \\
%&= 
%\int_{\cal D} \psi J(\Phi'(b_1),b_1) - \psi J(\Psi(b_1),b_1) - \Psi(b_1) J(\psi, q) + \frac{1}{2}h_1 J(\Psi(b_1),b_1) + \Psi(b_1)J(\psi,q) \,dx\,dy \\
%&= 0.
%\end{aligned}
%\end{equation}
%The integration by parts require the boundary conditions \eqref{eq:bdyTQGstream}.
%\end{proof}
The proofs that the integral quantities $H_{TQG}$ in \eqref{eq:erg-tqg} and $C_{\Phi,\Psi}$ in \eqref{eq:casimirstqg} are conserved by the TQG equations  in \eqref{eq:bTQGstream} and \eqref{eq:TQGinqform} are both direct calculations which invoke the boundary conditions \eqref{eq:bdyTQGstream}.
The integral conservation laws $C_{\Phi,\Psi}$ for the TQG equations have the same form as the family of integral conserved quantities for the TRSW equations, defined in \eqref{eq:casimirstrsw}, and likewise the integral conserved quantities \eqref{eq:casimirstL1} for the TL1 equations with the corresponding potential vorticity variable for TL1  in \eqref{def:pvtL1}. The persistence of these integral conservation laws for the deterministic TQG equations is best explained in terms of their Hamiltonian formulation.

\paragraph{Hamiltonian formulation of deterministic TQG}
The Hamiltonian form of the deterministic TQG equations in \eqref{eq:bTQGstream} and \eqref{eq:TQGinqform} is given by 
\begin{align}
\frac{dF}{dt} = \{ F, H\} = - \int 
\begin{bmatrix}
\delta F/ \delta q \\ \delta F/ \delta b
\end{bmatrix}^T
\begin{bmatrix}
J(\,q-b\,,\,\cdot\,) & J(\,b\,,\,\cdot\,)
\\
J(\,b\,,\,\cdot\,) & 0
\end{bmatrix}
\begin{bmatrix}
\delta H_{TQG}/ \delta q = -\,\psi \\ \delta H_{TQG}/ \delta b = -\,h_1/2
\end{bmatrix}
d^2x
\label{eqn:TQG-LPB}
\end{align}
for the energy Hamiltonian in equation \eqref{eq:erg-tqg}.
The Poisson matrix in \eqref{eqn:TQG-LPB} is a Poisson deformation of the Lie--Poisson Hamiltonian matrix by the invertible linear transformation $(q,b)\to (q-b,b)$. This deformation preserves the the integral conservation laws $C_{\Phi,\Psi}$ in equation \eqref{eq:casimirstqg} because their variational derivatives lie in the kernel of the Lie--Poisson Hamiltonian matrix operator. That is the Poisson bracket $\{C_{\Phi,\Psi},H\}$ vanishes for \emph{every} choice of functional $H(q,b)$, not just for the energy Hamiltonian $H_{TQG}(q,b)$ in  \eqref{eq:erg-tqg} which generates the deterministic TQG equations in \eqref{eq:bTQGstream} and \eqref{eq:TQGinqform} under the action of the Poisson bracket in \eqref{eqn:TQG-LPB}. The Hamiltonian structure allows one to investigate linear and nonlinear stability of the TQG model. This will be done in future work.

%\begin{align*}
%\{F,C_{\Phi,\Psi}\} &= -\frac{\delta F}{\delta q}\left(J\Big(q+b_1,\frac{\delta C_{\Phi,\Psi}}{\delta q}\Big) + J\Big(b_1,\frac{\delta C_{\Phi,\Psi}}{\delta b_1}\Big)\right) - \frac{\delta F}{\delta b_1}\left(J\Big(b_1,\frac{\delta C_{\Phi,\Psi}}{\delta q}\Big)\right)\\
%&= -\frac{\delta F}{\delta q}\left(J\Big(q+b_1,\Psi(b_1)\Big) + J\Big(b_1,\Phi'(b_1)+q\Psi'(b_1)\Big)\right) - \frac{\delta F}{\delta b_1}\left(J\Big(b_1,\Psi(b_1)\Big)\right)\\
%&= -\frac{\delta F}{\delta q}\left( J\Big(q,\Psi(b_1)\Big) + J\Big(b_1\Psi'(b_1),q\Big)\right)\\
%&= -\frac{\delta F}{\delta q}\left(J(q,\Psi(b_1) - J(q,\Psi(b_1)\right)\\
%&= 0
%\end{align*}

\paragraph{Hamiltonian formulation of stochastic TQG}
As with the TRSW equations, introducing stochasticity via the Euler--Poincar\'e theorem will preserve the conservation laws in \eqref{eq:casimirstqg}. However, introducing stochasticity would not preserve energy in \eqref{eq:erg-tqg}, because the stochastic Lagrangian depends explicitly on time through the Brownian motion. Nonetheless, as discussed in remark \ref{remark:Casimirs} for the stochastic TRSW equations,  the stochastic TQG equations may still possess a Hamiltonian formulation. By coupling the potential vorticity to noise we can construct the noise Hamiltonians as follows
\begin{equation}
H_i\circ dW_t^i = \int_{\cal D} q\,\vartheta_i(\bx) \,dx\,dy\circ dW_t^i,
\label{eq:noisehamTQG}
\end{equation}
where $\vartheta_i(\bx)$ are functions of space but not of time. Each $\vartheta_i(\bx)$ is associated with an independent Brownian motion $W_t^i$. For each $i$, we have $\nabla^\perp\vartheta_i(\bx) = \zh\times\nabla\vartheta_i(\bx) = \boldsymbol \xi_i(\bx)$. Thus, by definition, the divergence of $\boldsymbol \xi_i(\bx)$ vanishes for each $i$. By taking the sum of the Hamiltonian $H_{TQG}\,dt$ in \eqref{eq:erg-tqg} and the noise Hamiltonians $H_i\circ dW_t^i$ in \eqref{eq:noisehamTQG}, we obtain a semimartingale Hamiltonian. Inserting this augmented Hamiltonian into the Poisson bracket \eqref{eqn:TQG-LPB} yields the following stochastic TQG equations
\begin{equation}
\begin{aligned}
{\sf d}q + J(\psi,q-b)\,dt + J(\vartheta_i,q-b)\circ dW_t^i &= -\frac{1}{2} J\left(h_1,b\right)dt,\\
{\sf d}b + J(\psi,b)\,dt + J(\vartheta_i,b)\circ dW_t^i &= 0.
\end{aligned}
\label{eq:STQGham}
\end{equation}
Since we changed only the Hamiltonian to obtain the stochastic TQG equations, the conservation laws that correspond to \eqref{eq:STQGham} are immediate. We should remark that \eqref{eq:STQGham} are not the equations one would obtain by applying asymptotic analysis to the stochastic thermal L1 equations. Following the same methods as in the deterministic case leads to a stochastic set of equations with a much smaller family of integral quantities. Since adding SALT, besides the energy, preserves the conservation laws, proceeding via asymptotic analysis does not produce the SALT stochastic version of TQG.

\paragraph{Conservation laws for stochastic TQG.}
The stochastic TQG equations \eqref{eq:STQGham} do not conserve energy. This is because the Hamiltonian that generates the dynamics depends on time explicitly, due to the presence of the Brownian motions. However, equations \eqref{eq:STQGham}  do conserve the same set of integral quantities $C_{\Phi,\Psi}$, defined in \eqref{eq:casimirstqg}, since the Poisson bracket for the deterministic TQG equations and the stochastic TQG equations is the same.

\section{Conclusion and outlook}\label{sec:openprobs}

Our motivation in this paper has been to prepare the mathematical framework for our ongoing investigations of Stochastic Transport in Upper Ocean Dynamics (STUOD) by using the stochastic data assimilation algorithms developed and applied previously to determine the eigenvectors $\sym{\xi}_i(\bx)$ in the cases of the stochastic Euler fluid equation and the 2-layer stochastic QG model in \cite{cotter2018modelling, cotter2019numerically}. This framework has been established by deriving a sequence of realistic 2D models of Upper Ocean Dynamics with buoyancy effects by using nested asymptotic expansions with a shared stochastic variational structure. The process of developing these sequential derivations has also revealed several open mathematical problems at each level of approximation for these new nonlinear stochastic partial differential equations, as listed below.
\begin{enumerate}[(i)]
\item An extensive computational simulation study will be needed for classifying the solution behaviour of these new stochastic TRSW and TQG equations. This computational study has been left as a future step, after having established its efficacy in section \ref{sec:TQG-numerical example}. 
\item These computational simulations will be required in the calibration of the eigenvectors $\nabla^\perp\vartheta_i(\bx)=\boldsymbol \xi_i(\bx)$ for the noise in equation \eqref{eq:STQGham} and their subsequent use in the new framework for data calibration, uncertainty quantification and data assimilation using particle filters following the SALT algorithm developed in \cite{cotter2018modelling, cotter2019numerically}. This future simulation study will prepare these models for applications in the analysis of the observed detailed upper ocean dynamics as seen in Figure \ref{Fig:globcurrent1} . 
\item Investigation of the numerous potential effects of the horizontal buoyancy gradients appearing in the elliptic equation for the thermal L1 Lagrange multiplier $\bu$ \eqref{TL1-soln} has been left as an open mathematical problem for further analysis and computational simulation. 
\item The issue of well-posedness of these new nonlinear stochastic partial differential equations \eqref{EP-KN-trsw} for TRSW and \eqref{eq:STQGham} for TQG has also been left for future mathematical investigation. 
\item Finally, we recall that our derivation of the stochastic barotropic TQG balanced model in \eqref{eq:STQGham} has neglected the potentially important effects of baroclinic instabilities which tend to re-stratify the fluid. In particular, a future study with baroclinic TQG would extend the QG analysis of baroclinic instability of the currents around the Lofoten Basin given in \cite{isachsen2015baroclinic} to include thermal effects. For an in-depth discussion of baroclinic effects in comparison to balanced models, see \cite{callies2016role}.

\end{enumerate}

\subsection*{Acknowledgments}
We are grateful for constructive suggestions offered in discussions with C. J. Cotter, B. Chapron, D. Crisan, S.R. Ephrati, B. Fox-Kemper, R. Hu, O. Lang, J.M. Leahy, J. C. McWilliams, E. M\'emin, O. Street and S. Takao. We are also grateful to the European Space Agency for access to the Sentinel 3B satellite data shown in Figure \ref{Fig:globcurrent1}. During this work, DDH and WP were partially supported by ERC Synergy Grant 856408 - STUOD (Stochastic Transport in Upper Ocean Dynamics). EL was supported by EPSRC grant [grant number EP/L016613/1]. EL is also grateful for the warm hospitality shown to him  at the Imperial College London EPSRC Centre for Doctoral Training in the Mathematics of Planet Earth \url{mpecdt.org}.

\bibliographystyle{alpha}
\bibliography{biblio}

\end{document}